\documentclass[10pt]{article}
 
\usepackage{amsfonts, amsthm}
\usepackage{amssymb}
\usepackage{amsmath}  
\usepackage{graphicx, comment, color, hyperref}
\usepackage{authblk}
\usepackage{mathrsfs}
\usepackage{cleveref}

\usepackage{placeins}

\usepackage[a4paper]{geometry}
\newtheorem{theorem}{Theorem}
\newtheorem{corollary}[theorem]{Corollary}
\newtheorem{lemma}[theorem]{Lemma}

\newtheorem{prop}[theorem]{Proposition}

\theoremstyle{remark}
\newtheorem{remark}{Remark}

\def\reff#1{(\ref{#1})}

\usepackage{tikz}
\usepackage{tikzscale}
\usetikzlibrary{calc}
\usetikzlibrary{patterns}

\usepackage{subcaption}
\usepackage{blox}

\newcommand{\ket}[1]{|#1\rangle}
\newcommand{\bra}[1]{\langle#1|}

\newcommand{\tr}{\text{\rm tr}}
\newcommand{\E}{\mathbb{E}}

\newcommand{\Ent}{\text{\rm Ent}}
\newcommand{\HS}{\text{\rm HS}}

\newcommand{\cE}{\mathcal{E}}
\newcommand{\cD}{\mathcal{D}}
\newcommand{\cB}{\mathcal{B}}
\newcommand{\cH}{\mathcal{H}}

\newcommand{\cK}{\mathcal{K}}
\newcommand{\cI}{\mathcal{I}}
\newcommand{\dd}{\text{\rm{d}}}

\newcommand{\cL}{\mathcal{L}}
\newcommand{\cP}{\mathcal{P}}
\newcommand{\Var}{\text{\rm{Var}}}

\newcommand{\id}{\mathcal{I}}

\def\e{\mathrm{e}}
\def\eps{\varepsilon}
\def\II{\mathbb I}

\def\RR{\mathbb{R}}

\title{Quantum reverse hypercontractivity: its tensorization and application to strong converses}

\author[1]{Salman Beigi}
\affil[1]{\it \small School of Mathematics, Institute for Research in Fundamental Sciences (IPM), Tehran, Iran}

\author[2,3]{Nilanjana Datta}
\affil[2]{\it \small Department of Applied Mathematics and Theoretical Physics, Cambridge, UK}

\author[3]{Cambyse Rouz\'{e}}
\affil[3]{\it \small Statistical Laboratory, Centre for Mathematical Sciences, University of Cambridge, Cambridge, UK}

\begin{document}

\maketitle

\begin{abstract}
In this paper we develop the theory of quantum reverse hypercontractivity inequalities and show how they can be derived from log-Sobolev inequalities. Next we prove a generalization of the Stroock-Varopoulos inequality in the non-commutative setting which allows us to derive quantum hypercontractivity and reverse hypercontractivity  inequalities solely from $2$-log-Sobolev and $1$-log-Sobolev inequalities respectively. 
We then prove some tensorization-type results providing us with tools to prove hypercontractivity and reverse hypercontractivity not only for certain quantum superoperators but also for their tensor powers. Finally as an application of these results, we generalize a recent technique for proving strong converse bounds in information theory via reverse hypercontractivity inequalities to the quantum setting. We prove strong converse bounds for the problems of quantum hypothesis testing and classical-quantum channel coding based on the quantum reverse hypercontractivity inequalities that we derive. 

\end{abstract}

%******************************************************************************
\section{Introduction}

Let $\{T_t:\, t\geq 0\}$ be a continuous semigroup of stochastic maps (a Markov semigroup) with a unique stationary distribution $\pi$. 
Defining the $p$-norm, for $p\geq 1$, of a function $f$ by $\|f\|_p:=(\E |f|^p)^{1/p}$, where the expectation is with respect to $\pi$, a simple convexity-type argument verifies that 
$\|T_tf\|_p\leq \|f\|_p$.  That is, $T_t$, for all $t\geq 0$, is a contraction under $p$-norms. Since  $p\mapsto \|f\|_p$ is non-decreasing, a stronger contractivity inequality is the following: 
\begin{align}\label{eq:T-HC1}
\|T_tf\|_p\leq \|f\|_q,
\end{align}
for $1\leq q\le p$ and $t=t(p)$ an {increasing} function of $p$ satisfying $t(q)=0$. Thus an inequality of this form is called a \emph{hypercontractivity inequality}. Since $T_0$ equals the identity map, the inequality~\eqref{eq:T-HC1} for $p=q$ reduces to an equality. Thus its infinitesimal version around $t=0$ must also hold. This infinitesimal version is derived from the derivative of the left hand side of~\eqref{eq:T-HC1} and is called a \emph{$q$-log-Sobolev inequality}\footnote{For sake of brevity, we refrain from defining the phrases shown in italics throughout this introduction. Please refer to the main text and references therein for details}. Such an inequality involves two quantities: the \emph{entropy function} and the \emph{Dirichlet form}. A log-Sobolev inequality guarantees the existence of a positive constant, called a \emph{log-Sobolev constant}, up to which the entropy function is dominated by the Dirichlet form. Not only can one derive log-Sobolev inequalities from hypercontractivity ones, but a collection of the former inequalities can also be used to prove hypercontractivity inequalities through integration. Thus log-Sobolev inequalities and hypercontractivity inequalities are essentially equivalent. 

A fundamental tool in the theory of log-Sobolev inequalities is the \emph{Stroock-Varopoulos inequality}. This inequality enables us to compare the Dirichlet forms associated to different values of $q$, using which a log-Sobolev inequality for $q=2$ can be used to derive a log-Sobolev inequality for any $q$. Indeed, the Stroock-Varopoulos inequality allows us to derive a collection of log-Sobolev inequalities from a single one, from which hypercontractivity inequalities can be proven by integration. 

Hypercontractivity inequalities were first studied in the context of quantum field theory~\cite{G75b, N66, SH-K72}, but later found several important applications in different areas of mathematics, e.g., concentration of measure inequalities~\cite{BLM13, RS13}, transportation cost inequalities~\cite{GL10}, estimating the mixing times~\cite{DSC96}, analysis of Boolean functions~\cite{deWolf08} and information theory~\cite{AG76, KA12}.    
One of the main ingredients of most of these applications is the so called \emph{tensorization property}. It states that the hypercontractivity inequality
$$\|T_t^{\otimes n} f\|_p\leq \|f\|_q,$$
is satisfied for every $n\geq 1$ if and only if it holds for $n=1$. That is, the hypercontractivity of $T_t$ is equivalent to the hypercontractivity of its tensor powers. Proof of the tensorization property is not hard, and can be obtained using the multiplicativity of the operator $(q\to p)$-norm. Another proof, based on the equivalence of log-Sobolev and hypercontractivity inequalities, uses chain rule and the subadditivity of the entropy function.

Hypercontractivity inequalities can also be studied for $ p, q<1$. Although $\|\cdot\|_p$ for $p<1$ is not a norm, it satisfies the \emph{reverse Minkowski inequality} from which one can show that $\|T_tf\|_p\geq \|f\|_p$ when $p<1$. Thus it is natural to consider inequalities of the form~\eqref{eq:T-HC1} for $ p, q<1$ in the reverse direction. Such inequalities are called \emph{reverse hypercontractivity inequalities}. The theory of log-Sobolev inequalities for the range of $q<1$ is developed  similarly and can be used for proving reverse hypercontractivity inequalities as well~\cite{MOS12}.

\paragraph{Quantum hypercontractivity inequalities:} 
The theory of hypercontractivity and log-Sobolev inequalities in the quantum  (non-commutative) case has been developed  by Olkiewicz and Zegarlinski~\cite{OZ99}. Here the semigroup of stochastic maps is replaced by  
a semigroup of quantum superoperators (QMS) representing the time evolution of an open quantum system under the Markovian approximation in the Heisenberg picture.   
Kastoryano and Temme in~\cite{KT13} used log-Sobolev inequalities to estimate the mixing time of quantum Markov semigroups. The study of quantum reverse hypercontractivity was initiated in~\cite{CMT15}, where following~\cite{MOS12} some applications were discussed. For other applications of hypercontractivity inequalities in quantum information theory see~\cite{M12, DB14, MSFW}.

Due to the non-commutative features of quantum physics, hypercontractivity and log-Sobolev inequalities in the quantum case are much more complicated.  
Therefore, despite the apparent analogy with the classical (i.e.~commutative) case, several complications arise. In particular, one of the main drawbacks of the theory in the non-commutative case is the lack of a general quantum Stroock-Varopoulos inequality. {As mentioned above, such an inequality would allow one to derive hypercontractivity inequalities solely from a $2$-log-Sobolev inequality.} {Special cases of the quantum Stroock-Varopoulos inequality}, called \emph{regularity} and \emph{strong regularity} properties, were considered in the literature and proved for certain examples~\cite{OZ99, KT13}. The most general result in this direction is a proof of the strong regularity property for a wide class of quantum Markov semigroups obtained in~\cite{BarEID17}.

Even more problematic is the issue of tensorization. As mentioned before, the proof of the tensorization property in the commutative case is quite easy and can be done with at least two methods, yet none of them generalize to the non-commutative case; {(i) The superoperator norm is not multiplicative in general, and (ii) one cannot interpret the quantum conditional entropy as an average of an entropic quantity over a smaller system, which is a crucial aspect of the proof in the classical setting}. Thus far, the tensorization property has been proven only for a few special examples of quantum Markov semigroups. In particular, it was proven for the \emph{qubit} depolarizing semigroup in~\cite{MO10, KT13} and is generalized for all unital qubit semigroups in~\cite{King14}. Moreover, {in~\cite{TPK14, MSFW}} some techniques were developed for bounding the log-Sobolev constants associated to the tensor powers of quantum Markov semigroups, which can be considered as an intermediate resolution of the tensorization problem.  {We also refer to ~\cite{bardet2018hypercontractivity,BK16} for the theory of hypercontractivity and log-Sobolev inequalities for completely bounded norms. }

\subsection{Our results}
In this paper we first develop the theory of quantum \emph{reverse} hypercontractivity inequalities beyond the unital case. This is done almost in a manner analogous to the (forward) hypercontractivity inequalities. Here, in contrast to~\cite{OZ99, KT13}, we need to use different normalizations for the entropy function as well as the Dirichlet form to make them non-negative even for parameters $p<1$. Our results in this part are summarized in Theorem~\ref{thm:gen-hyper}.

Our next result is a quantum Stroock-Varopoulos inequality for both the forward and reverse cases. We prove this inequality under the assumption of \emph{strong reversibility} of the QMS. We provide two proofs for the quantum Stroock-Varopoulos inequality. The first proof is based on ideas in~\cite{CM16} {and \cite{OZ99}}. The second proof is based on ideas in~\cite{BarEID17} in which the strong regularity is proven under the same assumption. Indeed, our quantum Stroock-Varopoulos inequality is a generalization of the strong regularity property established in~\cite{BarEID17}. Theorem~\ref{thm:QSVineq} states our result in this part.

We then prove some tensorization-type results. The first one, Theorem~\ref{thm:tensorization-alpha-1}, provides a uniform bound on the $1$-log-Sobolev constant of  \emph{generalized depolarizing semigroups} and their tensor powers. The proof of this result is a generalization of the proof of a similar result in the classical case~\cite{MOS12}. This tensorization result together with our Stroock-Varopoulos inequality gives a reverse hypercontractivity inequality which is used in the subsequent section. The second tensorization result, Theorem~\ref{thm:tensorization-alpha-2}, shows that the $2$-log-Sobolev constant of the $n$-fold tensor power of a \emph{qubit} generalized depolarizing semigroup is independent of $n$. Next, in Theorem~\ref{thm:LSC-2-simple} we explicitly compute this $2$-log-Sobolev constant. Finally, in Corollary~\ref{cor:2-LSC-general} we use these results to establish a uniform bound on the $2$-log-Sobolev constant of \emph{any} qubit quantum Markov semigroup and its tensor powers. We note that the latter bound improves over the bounds provided in~\cite{TPK14}.

Let us briefly explain the ideas behind the latter tensorization results. Previously, Theorem~\ref{thm:tensorization-alpha-2} was known in the unital case (the usual depolarizing semigroup), the proof of which was based on an inequality on the norms of a $2\times 2$ block matrix and its submatrices from~\cite{King03}. Our proof of Theorem~\ref{thm:tensorization-alpha-2} is based on the same inequality. First in Lemma~\ref{lem:entropy-inequality} we derive an infinitesimal version of that inequality in terms of the entropies of a $2\times 2$ block matrix and its submatrices, and then use it to prove Theorem~\ref{thm:tensorization-alpha-2}. To prove Theorem~\ref{thm:LSC-2-simple} we need to show that a certain function of qubit density matrices  is optimized over diagonal ones. Once we show this, the explicit expression for the $2$-log-Sobolev constant is obtained from the associated classical log-Sobolev constant derived in~\cite{DSC96}. Finally, Corollary~\ref{cor:2-LSC-general} is a quantum generalization of a classical result from~\cite{DSC96} with an essentially similar proof except that we should take care of tensorization separately.

Finally, we apply the quantum reverse hypercontractivity in proving strong converse bounds for the tasks of quantum hypothesis testing and classical-quantum channel coding. In the next section, we briefly explain the key idea behind the application of reverse hypercontractivity to the problem of classical hypothesis testing.

\subsection{Application to hypothesis testing problem}
Recently, the authors of~\cite{LHV17} introduced a new technique to prove strong converse results in information theory using reverse hypercontractivity inequalities. In the following we briefly explain the ideas via the problem of hypothesis testing.  

Suppose that $n$ samples  independently drawn from a probability distribution on some sample space $\Omega$ are provided, and the task is to distinguish between two possible hypotheses which are given by the distributions $P$ and $Q$ on $\Omega$. In this setting, we apply a test function\footnote{The test could be probabilistic, but for simplicity of presentation we restrict to deterministic tests.} $f:\Omega^n\to \{0, 1\}$ to make the decision; Letting $(x_1, \dots, x_n)\in \Omega^n$ be the observed samples, if $f(x_1, \dots, x_n)$ equals $1$, we infer the hypothesis to be $P$, and otherwise infer it to be $Q$. The following two types of error may occur: the error of Type~I of wrongly inferring the distribution to be $Q$ given by $\alpha_n(f):=P^{\otimes n}(f=0)$, and the error of Type~II of wrongly inferring the distribution to be $P$ given by $\beta_n(f):=Q^{\otimes n}(f=1)$. In the \emph{asymmetric} regime, we further assume that $\alpha_n(f)$ is uniformly bounded by some fixed error $\eps\in (0,1)$, and we are interested in the smallest possible achievable error $\beta_n(f)$.  

The idea in~\cite{LHV17} is to use the following variational formula for the relative entropy between $P$ and $Q$ (see, e.g.,~\cite{RS13}):
\begin{align}\label{eq5}
nD(P\|Q)=D(P^{\otimes n}\| Q^{\otimes n})=\sup_{g>0} \E_{P^{\otimes n}}[\log g]-\log \E_{Q^{\otimes n}}[g], 
\end{align}
where $\E_{P^{\otimes n}}$ stands for the expectation with respect to the distribution $P^{\otimes n}$, and the maximum is over functions $g$ on $\Omega^n$. This formula is indeed used for $g$ being a \emph{noisy version} of $f$. To get this noisy version  a Markov semigroup is employed. 

For any function $h:\Omega\to \mathbb R$ define
\begin{align}\label{eq:Tt-simple}
	T_t(h):= \e^{-t}h +(1-\e^{-t}) \E_P[h],
	\end{align}
	{These maps define a classical version of the generalized quantum depolarizing semigroup (see \Cref{eq:def-Phi}). }
That is, for every $x\in \Omega$, we have $T_t(h)(x) = \e^{-t}h(x) + (1-\e^{-t}) \E_P[h]$. Then $\{T_t:\, t\geq 0\}$ forms a semigroup that satisfies the following reverse hypercontractivity inequality~\cite{MOS12}:
\begin{align}\label{eq1}
	\|T_t(h)\|_{q}\ge \|h\|_{p},   \qquad \quad \forall p, q, t , \quad 0\leq q<p<1,~~~ t\ge \log\left(  \frac{1-q}{1-p}\right),
\end{align}	
where the norms are defined with respect to the distribution $P$, i.e., $\|h\|_p = \big( \E_P[|h|^p] \big)^{1/p}$.
Now the idea is to use~\eqref{eq5} for $g=T_t^{\otimes n} f$ as follows:
\begin{align}\label{eq:nD-Tt-f}
nD(P\| Q) \geq \E_{P^{\otimes n}}[\log T_t^{\otimes n}f]-\log \E_{Q^{\otimes n}}[T_t^{\otimes n} f].
\end{align}  
Bounding the second term on the right hand side is easy. Letting $\gamma=\left\|  \frac{dP}{dQ} \right\|_\infty$ we have
\begin{align}
	\E_{Q^{\otimes n}}[T_t^{\otimes n}(f)]&=\E_{Q^{\otimes n}}\big[\big(\e^{-t}+(1-\e^{-t})\E_P  \big)^{\otimes n}f  \big]\nonumber\\
	&\le \E_{Q^{\otimes n}}\big[\big(\e^{-t}+\gamma(1-\e^{-t})\E_Q  \big)^{\otimes n}f  \big]\nonumber\\
	&= \left(\e^{-t}+\gamma(1-\e^{-t})\right)^n \E_{Q^{\otimes n}} [f]\nonumber\\
	&= \left(\e^{-t}+\gamma(1-\e^{-t})\right)^n \beta_n(f)\nonumber\\
	&\le \e^{\left(\gamma-1\right)nt} \beta_n(f),\label{eq7}
\end{align}	
where the last inequality follows from $\e^{\gamma t}-1\ge \gamma(\e^t-1)$ for $\gamma\geq 1$.

Now we need to bound the first term in terms of $\alpha_n(f)$. The crucial observation here is that 
\begin{align}\label{eq:norm-0-ln}
\|h\|_0 = \lim_{r\rightarrow 0} \|h\|_r = \e^{\E_{{P}}[\log |h|]}.
\end{align}
It is then natural to use the reverse hypercontractivity inequality~\eqref{eq1} for $q=0$. In fact, using the tensorization property, that~\eqref{eq1} also holds for $T_t^{\otimes n}$, we have
\begin{align}
	\E_{P^{\otimes n}}[\log T_t f]&=\log \|T_t^{\otimes n}(f)\|_{0}\nonumber\\
	&\ge \log \|f\|_{1-\e^{-t}}\nonumber\\
	&\ge\frac{1}{1-\e^{-t}}\log \E_{P^{\otimes n}}[f]\nonumber\\
	&\ge \left(\frac{1}{t}+1\right)\log (1-\alpha_n(f)),\label{eq6}
\end{align}	
where the second line follows from the reverse hypercontractivity inequality, the third line follows from the fact that $T_t^{\otimes n}(f)$ takes values in $[0,1]$, and the last line follows from $\e^{-t}\ge 1-t$. 
Now using~\eqref{eq7} and~\eqref{eq6} in~\eqref{eq:nD-Tt-f}, using $\alpha_n(f)\leq \eps$ and optimizing over the choice of $t> 0$ we arrive at
\begin{align}\label{eq8}
	\beta_n(f)\ge (1-\eps)\e^{-nD(P\|Q) -2\sqrt{n \left\| \frac{dP}{dQ} \right\|_\infty  \log\frac{1}{1-\eps} } }.
\end{align}

In the present work, we show that the above analysis can be carried over to the quantum setting. Let us explain the similarities with the classical case as well as difficulties we face in doing this. Firstly, a variational expression for the quantum relative entropy similar to~\eqref{eq5} is already known~\cite{Petz88}. Secondly, the semigroup~\eqref{eq:Tt-simple} is easily generalized to the generalized depolarizing semigroup in the quantum case. Thirdly, the reverse hypercontractivity inequality~\eqref{eq1} is derived in the quantum case from our theory of quantum reverse hypercontractivity as well as our quantum Stroock-Varopoulos inequality. However we need this inequality in its $n$-fold tensor product form, for which we use our tensorization-type result.  Also, generalizing the computations in~\eqref{eq7} to the quantum case is straightforward. Nevertheless, we face a  problem in the next step; The crucial identity~\eqref{eq:norm-0-ln} no longer holds in the non-commutative case. {Indeed, as far as we know, non-commutative $L_p$-norms do not possess a closed expression in the limit $p\to0$.} To get around this problem, instead of a variational formula similar to~\eqref{eq5}, we use our quantum reverse hypercontractivity inequality together with a variational formula for $p$-norms (obtained from the \emph{reverse H\"older inequality}). Then we derive an inequality of the form~\eqref{eq8} by taking an appropriate limit. 

Section~\ref{sec6} contains our results on applications of reverse hypercontractivity inequalities to strong converse of the quantum hypothesis testing as well as the classical-quantum channel coding problems.

%******************************************************************************
\section{Notations}\label{sec:notation}
For a Hilbert space $\cH$, the algebra of (bounded) linear operators acting on $\cH$ is denoted by $\cB(\cH)$. The adjoint of $X\in \cB(\cH)$ is denoted by $X^\dagger$ and $$|X|:=\sqrt{X^\dagger X}.$$
The subspace of self-adjoint operators is denoted by $\cB_{sa}(\cH) \subset \cB(\cH)$. When $X\in \cB_{sa}(\cH)$ is positive semi-definite (positive definite) we represent it by $X\geq 0$ ($X> 0$). We
let $\cP(\cH)$ be the cone of positive semi-definite operators on $\cH$ and $\cP_{+}(\cH) \subset  \cP(\cH)$ the set of (strictly) positive operators. Further, let $\cD(\cH):=\lbrace\rho\in\cP(\cH)\mid \tr\rho=1\rbrace$ denote the set of density operators (or states) on $\cH$, and $\cD_+(\cH):=\cD(\cH)\cap \cP_+(\cH)$ denote the subset of \emph{faithful} states. We denote the support of an operator $A$ by ${\mathrm{supp}}(A)$. We let $\mathbb{I}\in\cB(\cH)$ be the identity operator on $\cH$, and $\id:\cB(\cH)\mapsto \cB(\cH)$ be the identity superoperator acting on~$\cB(\cH)$.

We sometimes deal with tensor products of Hilbert spaces. In this case, 
in order to keep track of subsystems, it is appropriate to label the Hilbert spaces as $\cH_A, \cH_B$ etc. We also denote $\cH_A\otimes \cH_B$ by $\cH_{AB}$. Then the subscript in $X_{AB}$ indicates that it belongs to $\cB(\cH_{AB})$. We also use  
$\cH^{\otimes n} = \cH_{A_1}\otimes \cdots \otimes \cH_{A_n}$ where $\cH_{A_i}$'s are isomorphic Hilbert spaces. Moreover, for any $S\subseteq \{1, \dots, n\}$ we use the shorthand notations  $A_S{=A^S=\{A_j: \, j\in S \}}$, and $\cH_{A_S}$ for $\bigotimes_{j\in S}\cH_{A_j}$. We also identify $A_{\{1, \dots, n\}}$ with $A^n$.

A superoperator $\Phi:\cB(\cH)\rightarrow \cB(\cH)$ is called \emph{positive} if $\Phi(X)\geq 0$ whenever $X\geq 0$. It is called \emph{completely positive} if $\id\otimes \Phi$ is positive where $\id:\cB(\cH')\rightarrow \cB(\cH')$ is the identity superoperator associated to an arbitrary Hilbert space $\cH'$. 
Observe that a positive superoperator $\Phi$ is hermitian-preserving meaning that $\Phi(X^\dagger) =\Phi(X)^\dagger$.
A superoperator is called \emph{unital} if $\Phi(\II)=\II$, and is called trace-preserving if $\tr \,\Phi(X)=\tr X$ for all $X$. The adjoint of $\Phi$, denoted by $\Phi^*$ is defined with respect to the Hilbert-Schmidt inner product:
\begin{align}\label{eq:adj-HS}
\tr\left(X^\dagger\Phi(Y)\right) = \tr\left(\Phi^*(X)^\dagger Y\right).
\end{align}
Note that the adjoint of a unital map is trace-preserving and vice versa.

%******************************************************************************
\subsection{Non-commutative weighted $L_p$-spaces}
Throughout the paper we fix $\sigma\in \mathcal D_+(\cH)$ to be a positive definite density matrix. We define 
$$\Gamma_\sigma(X):= \sigma^{\frac12}X\sigma^{\frac12}.$$
Then $\cB(\cH)$ is equipped with the inner product
$$\langle X, Y\rangle_\sigma:= \tr\left(X^\dagger \Gamma_\sigma(Y)\right)= \tr\left(\Gamma_\sigma(X^\dagger)Y\right).$$ 
Note that if $X, Y\geq 0$ then $\langle X, Y\rangle_\sigma\geq 0$. 
This inner product induces a norm on~$\cB(\cH)$:
\begin{align}\label{eq:2-norm-0}
\|X\|_{2, \sigma} := \sqrt{\langle X, X\rangle_\sigma}.
\end{align}
This $2$-norm can be generalized for other values of $p$. 
For every $p\in \mathbb R\setminus\{0\}$ we define

\begin{align}\label{normp}
	\|X\|_{p, \sigma}:= \tr\left[\big|\Gamma_\sigma^{\frac1p}(X)\big|^p\right]^{\frac1p} = \tr\left[\big|\sigma^{\frac1{2p}}X\sigma^{\frac1{2p}}\big|^p\right]^{\frac1{p}}\equiv \big\|\Gamma^{\frac{1}{p}}_\sigma(X)\big\|_p  ,
	\end{align}
where 
$$\|X\|_p:=\left( \tr\,|X|^p\right)^{1/p},$$ 
denotes the (generalized) Schatten norm of order $p$. In particular, if $X> 0$ then  $\|X\|_{p, \sigma}^p=\tr\big[\Gamma_\sigma^{1/p}(X)^p\big]$. Note that this definition reduces to~\eqref{eq:2-norm-0} when $p=2$.
The values of $\|X\|_{p, \sigma}$ for $p\in \{0, \pm\infty\}$ are defined in the limits. {Since the function $p\mapsto \|X\|_{p,\sigma}$ is increasing and bounded below by $0$, by the monotone convergence theorem, the limit $p\to0$ exists but does not have a closed expression, as opposed to the classical setting (cf \Cref{eq:norm-0-ln}).}
Observe {also} that $\|X\|_{p, \sigma} = \|X^{\dagger}\|_{p, \sigma}$ for all $X$. Moreover, $\|\cdot\|_{p, \sigma}$ for $1\leq p\leq \infty$ satisfies the triangle inequality (the Minkowski inequality) and is a norm. The dual of this norm is $\|\cdot\|_{\hat p, \sigma}$ where $\hat p$ is the H\"older conjugate of $p$ given by
\begin{align}\label{eq:holder-conj}
\frac{1}{p} +\frac{1}{\hat p}=1\,,
\end{align}
{where $p>1$, and $\hat{p}=+\infty$ for $p=1$.} We indeed for $1\leq p\leq \infty$ and arbitrary $X$ have~\cite{OZ99} 
\begin{align}\label{eq:holder-1}
 \|X\|_{p, \sigma} = \sup_Y \frac{|\langle X, Y\rangle_\sigma|}{\|Y\|_{\hat p, \sigma}}.
\end{align}
Moreover, for $-\infty< p<1$, $p\neq 0$ and \emph{positive definite} $X$ we have
 \begin{align}\label{eq:holder-2}
\|X\|_{p, \sigma}= \inf_{Y>0} \frac{\langle X, Y\rangle_\sigma}{\|Y\|_{\hat p, \sigma}}\,,
 \end{align}
{where again $\hat p$ is defined via~\eqref{eq:holder-conj}.\footnote{In the case $p=0$, we define $\hat{p}=0$ (see e.g., Definition 1.2 of \cite{MOS12})}. This identity is a consequence of the \emph{reverse H\"older inequality}:
\begin{lemma}[Reverse H\"{o}lder inequality]\label{RHolder}
	Let $X\ge0$ and $Y>0$. Then, for any  $p< 1$ with H\"{o}lder conjugate $\hat{p}$ we have
	\begin{align*}
		\langle X,Y\rangle_\sigma \ge \|X\|_{p,\sigma}\|Y\|_{\hat{p},\sigma}.
	\end{align*}
\end{lemma}

\begin{proof}
	The proof is a direct generalization of equation (32) of \cite{TBH14} (see also Lemma 5 of \cite{CMT15}): for any $A\ge 0$ and $B>0$, 
	\begin{align*}
		\tr(AB)\ge \|A\|_p\|B\|_{\hat{p}}.
	\end{align*}
	From there, choosing $A:=\Gamma_{\sigma}^{\frac{1}{p}}(X)$ and $B:=\Gamma_\sigma^{\frac{1}{\hat{p}}}(Y)$, 
	\begin{align*}
		\langle X,Y\rangle_{\sigma}=\tr\big(\sigma^{1/p}X\sigma^{1/p}\sigma^{1/\hat{p}}Y\sigma^{1/\hat{p}}\big)=\tr(AB)\ge \|A\|_p\|B\|_{\hat{p}}=\|X\|_{p,\sigma}\|Y\|_{\hat{p},\sigma}.
	\end{align*}
	
\end{proof}	

Another property of $\|\cdot\|_{p, \sigma}$ for $-\infty\leq p<1$ is the \emph{reverse Minkowski inequality}. As mentioned above, when $p\geq 1$, the triangle inequality is satisfied due to the Minkowski inequality. When $p<1$ we have the inequality in the reverse direction:
$$\|X\|_{p, \sigma} + \|Y\|_{p, \sigma}\leq \|X+Y\|_{p, \sigma}.$$ 
Again this inequality in the special case of $\sigma$ being the completely mixed state is proven in~\cite{CMT15} but the generalization to arbitrary $\sigma$ is immediate.

For arbitrary $p, q$ define the \emph{power operator} by
$$I_{q, p}(X) := \Gamma_\sigma^{-\frac1q }\left( \big|\Gamma_\sigma^{\frac 1 p}(X)\big|^{\frac pq}\right).$$
Here are some immediate properties of the power operator. 

\begin{prop}{\rm\cite{OZ99, KT13}} For all $q,r,p\in(-\infty,\infty)\backslash\{0\}$ and $X\in\cB(\cH)$:
\begin{enumerate}
\item[{\rm(i)}] $\|I_{q, p}(X)\|_{q, \sigma}^q =\|X\|_{p, \sigma}^p$.  In particular we have $\|I_{p, p}(X)\|_{p, \sigma} = \|X\|_{p, \sigma}$.
\item[{\rm(ii)}] $I_{q, r}\circ I_{r, p} = I_{q, p}$.
\item[{\rm(iii)}] For $X\geq 0$ we have $I_{p, p}(X)=X$. 
\end{enumerate}
\label{prop:power}
\end{prop}

%******************************************************************************
\subsection{Entropy}
For a given $\sigma\in \cD_+(\cH)$ and arbitrary $p\neq 0$ we define the~\emph{entropy} function\footnote{Our entropy function here is different from the one in~\cite{KT13} by a factor of $p$. This modification ensures us that if $X$ and $\sigma$ commute, we get the usual entropy function in the classical case. Moreover, this extra factor makes the entropy function non-negative even for $p<0$.} for $X> 0$ by
$$\Ent_{p, \sigma}(X):= \tr\Big[   
\big(\Gamma_\sigma^{\frac 1p}(X)\big)^p\cdot \log \big(\Gamma_\sigma^{\frac 1p}(X)\big)^p    \Big] -\tr\Big[\big(\Gamma_\sigma^{\frac 1p}(X)\big)^p\cdot \log \sigma\Big]- \|X\|_{p, \sigma}^p\cdot \log \|X\|_{p, \sigma}^p.$$
As usual, the entropy function for $p\in \{0, \pm\infty\}$ is defined in the limit. 

\begin{remark}
	When $p> 0$, in the definition of the entropy we can take $X$ to be positive semi-definite. However, when $p<0$, we need to consider $X$ to be positive definite in order to avoid difficulties. For this reason, in the rest of the paper we state our definitions and results for positive definite $X$, keeping in mind that when $p, q>0 $ they can easily be generalized to positive semi-definite $X$ (say, by taking an appropriate limit).
\end{remark}

The significance of the entropy function comes from its relation to the derivative of the $p$-norm.

\begin{prop}{\rm\cite{OZ99, KT13}}
 For a differentiable operator valued function $p\mapsto X_p$ we have, for any $p\in\RR\backslash\{0\}$:
$$\frac{\dd}{\dd p}\|X_p\|_{p,\sigma} = \frac{1}{p^2}\|X_p\|_{p, \sigma}^{1-p}\cdot \left( \frac{1}{2}\Ent_{p, \sigma}\big(I_{p, p}(X_p)\big) +\frac{1}{2}\Ent_{p, \sigma}\big(I_{p, p}(X_p^{\dagger})\big)  + \gamma   \right).$$
Here $\gamma$ is given by 
$$\gamma=\frac{p^2}{2}\left( \tr\Big[\Gamma_\sigma^{\frac 1p}(Z_p^{\dagger})\cdot \Gamma_\sigma^{\frac 1 p}(X_p)\cdot \big| \Gamma_\sigma^{\frac 1p}(X_p) \big|^{p-2}\Big] +\tr\Big[\Gamma_\sigma^{\frac 1p}(X_p^{\dagger})\cdot \Gamma_\sigma^{\frac 1 p}(Z_p)\cdot \big| \Gamma_\sigma^{\frac 1p}(X_p) \big|^{p-2}\Big]\right),$$
where 
$Z_p := \frac{\dd}{\dd p}X_p$.
\label{prop:norm-derivative}
\end{prop}

We will be using two special cases of this proposition. First, if $X_p> 0$ for all $p$, we have 
$$\frac{\dd}{\dd p}\|X_p\|_{p,\sigma} = \frac{1}{p^2}\|X_p\|_{p, \sigma}^{1-p}\cdot \left( \Ent_{p, \sigma}(X_p)   + p^2\tr\Big[  \Gamma_\sigma^{\frac 1p}(Z_p) \cdot \Gamma_\sigma^{\frac 1p} (X_p)^{p-1} \Big]   \right).$$
Second, if $X_p=X$ is independent of $p$ we have
\begin{align}\label{eq:norm-derivative-ind}
\frac{\dd}{\dd p}\|X\|_{p,\sigma} = \frac{1}{p^2}\|X\|_{p, \sigma}^{1-p}\cdot \left( \frac{1}{2}\Ent_{p, \sigma}\big(I_{p, p}(X)\big) +\frac{1}{2}\Ent_{p, \sigma}\big(I_{p, p}(X^{\dagger})\big)    \right).
\end{align}

We will also use the following properties of the entropy function that are easy to verify. 

\begin{prop}{\rm \cite{KT13}}
\begin{enumerate}
\item[{\rm (i)}] $\Ent_{p, \sigma}(I_{p, 2}(X)) = \Ent_{q, \sigma}(I_{q, 2}(X))$ for all $p, q\in \RR\backslash\{0\}$ and $X\in\cB(\cH)$. 

\item[{\rm (ii)}] $\Ent_{p, \sigma}(cX) = c^p \Ent_{p, \sigma}(X)$ for all $X> 0$ and constants $c> 0$.

\item[{\rm (iii)}] For any density matrix $\rho$ we have 
$$\Ent_{2, \sigma}\big(\Gamma_\sigma^{-\frac 1 2}(\sqrt \rho)\big) =  D(\rho\|\sigma),$$ 
where $D(\rho\|\sigma) = \tr(\rho\log \rho) - \tr(\rho\log \sigma)$ is Umegaki's relative entropy.
\item[{\rm (iv)}] For any density matrix $\rho$ we have
$$\Ent_{1, \sigma}\big(\Gamma_\sigma^{-1}(\rho)\big) = D(\rho\| \sigma).$$
\end{enumerate}
\label{prop:ent}
\end{prop}

\begin{corollary}\label{corol:ent-positive}
\begin{enumerate}
\item[{\rm (a)}] For all $X>0$ and arbitrary $p\in\RR\backslash\{0\}$ we have $\Ent_{p, \sigma}(X)\geq 0$.
\item[{\rm (b)}] For all $X>0$, the map $p\mapsto \|X\|_{p, \sigma}$ is non-decreasing on $\RR$. 
\item[{\rm (c)}] $X\mapsto \Ent_{1, \sigma}(X)$ is a convex function on positive semi-definite matrices.
\end{enumerate}
\end{corollary}

\begin{proof}
(a) By part (i) of the previous proposition it suffices to prove the corollary for $p=1$. Moreover, by part (ii) we may assume that $X$ is of the form $X=\Gamma_\sigma^{-1}(\rho)$ for some density matrix $\rho$. Then by part (iv) we have
$\Ent_{1, \sigma}(X) = D(\rho\| \sigma)\geq 0$.

\noindent
(b) By (a) both $\Ent_{p, \sigma}(I_{p,p}(X))$ and $\Ent_{p, \sigma}(I_{p,p}(X^\dagger))$ are non-negative. Thus using~\eqref{eq:norm-derivative-ind} the derivative of $p\mapsto \|X\|_{p, \sigma}$ is non-negative, and this function is non-decreasing.

\noindent 
(c) {This is a direct consequence of the joint convexity of $(\rho,\sigma)\mapsto D(\rho\|\sigma)$ {(see e.g., \cite{wolftour})}.}
%(c) Given $X, Y\geq 0$ define 
%$$f(p) = \frac{1}{2}\big( \|X\|_{p, \sigma} + \|Y\|_{p, \sigma}\big)-\Big\|\frac 12(X+Y)\Big\|_{p, \sigma}.$$
%Using Proposition~\ref{prop:norm-derivative} we have
%$$f'(1) =  \frac 12\big(\Ent_{1, \sigma}(X)+\Ent_{1, \sigma}(Y)\big) - \Ent_{1, \sigma}\big(\frac 12(X+Y)\big).$$
%On the other hand by the Minkowski inequality $f(p)\geq 0$ for all $p\geq 1$. Moreover, $f(1)=0$. Therefore, $f'(1)\geq 0$ which gives the desired result. 

\end{proof}

%******************************************************************************
\subsection{Quantum Markov semigroups}
A quantum Markov semigroup (QMS) is the basic model for the evolution of an open quantum system in the Markovian regime. Such quantum Markov semigroup
 (in the Heisenberg picture) is a set $\{\Phi_t:\, t\geq 0\}$ of completely positive unital superoperators $\Phi_t: \cB(\cH)\rightarrow \cB(\cH)$ of the form 
$$\Phi_t = \e^{-t\cL},$$
where $\cL: \cB(\cH)\rightarrow \cB(\cH)$ is a superoperator called the Lindblad generator of the QMS. The general form of such a Lindblad generator is characterized in~\cite{Lind, GKS76}.
We note that $\Phi_0=\id$ and $\Phi_{t+s}=\Phi_s\circ\Phi_t$. Moreover, for any $X\in\cB(\cH)$ we have
$$\frac{\dd}{\dd t}\Phi_t(X) = -\cL\circ \Phi_t(X) = -\Phi_t\circ \cL(X).$$
In particular, since $\Phi_t$ is assumed to be unital, we have
$$\cL(\II)=0.$$

The dual of $\cL$ generates the associated QMS in the Schr\"odinger picture:
$\Phi_t^* = \e^{-t\cL^*}$ where $\cL^*$ is the adjoint of $\cL$ with respect to the Hilbert-Schmidt inner product defined in~\eqref{eq:adj-HS}. Since $\cL$ is not full-rank, there exists some non-zero $\sigma$ in the kernel of $\cL^*$ as well. Then $\sigma$ is an invariant of the semigroup $\{\Phi_t^*: t\geq 0\}$, i.e., $\Phi_t^*(\sigma) = \sigma$ for all $t\geq 0$. 
Throughout the paper we assume that such a $\sigma$ is \emph{unique} (up to scaling) and \emph{full-rank}. Then it can be proven that $\sigma$ is a density matrix.\footnote{By Brouwer's fixed-point theorem, $\Phi_1^*$, has a fixed point in $\cD(\cH)$ because it maps this compact convex set to itself. On the other hand, since $\Phi_t^*= (\Phi_1^*)^t$, any fixed point of $\Phi_1^*$ is an invariant of the whole semigroup. Thus $\{\Phi_t^*:\, t\geq 0\}$ always has an invariant state in $\cD(\cH)$.} Thus by the above uniqueness and full-rankness assumptions, $\{\Phi_t^*:\, t\geq 0\}$ admits a unique invariant state $\sigma$ in $\cD_+(\cH)$. We call such a QMS \emph{primitive}. 
Observe that for a primitive QMS the identity operator $\II$ is the unique (up to scaling) element in the kernel of $\cL$.

We say that the QMS is $\sigma$-\emph{reversible} or satisfies the \emph{detailed balanced} condition with respect to some $\sigma\in \cD_+(\cH)$ if
$$\Gamma_\sigma\circ \cL\circ \Gamma_\sigma^{-1}= \cL^*.$$
From this equation and $\cL(\II)=0$ it is clear that 
$$\cL^*(\sigma)=0,$$
and that $\sigma$ is a fixed point of $\Phi_t^*$. Therefore, if the QMS is primitive and $\sigma$-reversible, then $\sigma$ would be the unique invariant state of $\{\Phi_t^*:\, t\geq 0\}$.

We will frequently use the following immediate consequence of reversibility.

\begin{lemma}\label{lem:reversible}
 $\cL$ is $\sigma$-reversible if and only if both $\cL$ and $\Phi_t$ are self-adjoint with respect to the inner product $\langle \cdot, \cdot\rangle_\sigma$, which means that for all $X, Y\in\cB(\cH)$ we have
$$\langle X, \cL(Y)\rangle_\sigma = \langle \cL(X), Y\rangle_\sigma, \qquad \langle X, \Phi_t(Y)\rangle_\sigma = \langle \Phi_t(X), Y\rangle_\sigma. $$
\end{lemma}

A primitive QMS with the unique invariant state $\sigma\in \cD_+(\cH)$ is called $p$-\emph{contractive} if it is a contraction under the $p$-norm, that is, for all $t\geq 0$ and $X> 0$ we have
$$\|\Phi_t(X)\|_{p,\sigma}\leq \|X\|_{p, \sigma}, \quad \text{if }\quad  p\geq 1.$$
It is called \textit{reverse $p$-contractive} if for all $t\ge 0$ and $X>0$
$$\|\Phi_t(X)\|_{p,\sigma}\geq \|X\|_{p, \sigma}, \quad \text{if}\quad p< 1.$$
We say that the QMS is contractive if it is $p$-contractive for all $p\ge 1$ and reverse $p$-contractive for $p<1$.

Two remarks are in line. Firstly, as mentioned before, when $p>0$ in the above definition we may safely take $X\geq 0$ (instead of $X>0$). For uniformity of presentation we prefer to take $X>0$ in order to jointly consider the cases $p>0$ and $p\leq 0$ in the definitions. Of course in the former case by taking an appropriate limit, a contractivity inequality for $X\geq 0$ can be derived once we have one for $X>0$. Secondly, in the above definition we restrict to positive definite (or positive semidefinite) $X$ since here $\Phi_t$ is a completely positive map, and the superoperator norm of completely positive maps (at least for $p\geq 1$) is optimized over positive semidefinite operators (see e.g.~\cite{DJKR16} and reference therein). {The proof of the following proposition is postponed to \Cref{app:contraction}}.

\begin{prop}\label{prop:contraction}
\begin{enumerate}
\item[{\rm (i)}] Any primitive QMS is (reverse) $p$-contractive for $p\in (-\infty, -1]\cup [1/2, +\infty)$.
\item[{\rm (ii)}] Any primitive QMS whose unique invariant state is $\sigma=\II/d$, the completely mixed state, is (reverse) $p$-contractive for all $p$.  
\end{enumerate}
\end{prop}

The reader familiar with the notion of \emph{sandwiched $p$-R\'enyi divergence}~\cite{MDSFT13, WWY14} would notice that $p$-contractivity is related to~\cite{Beigi13} the data processing inequality of sandwiched $p$-R\'enyi divergences, which is known to hold~\cite{FL13, Beigi13, MDSFT13} for $p\geq 1/2$. 
In Appendix~\ref{app:contraction} we give a proof of part (i) for the range $p\in (-\infty, -1]\cup [1/2, 1)$ based on new ideas which may be of independent interest. Moreover, later in Corollary~\ref{cor:NN-Dirichlet}, under a stronger assumption than primitivity we will prove (reverse) $p$-contractivity for all $p$.

An important example of classical semigroups is generated by the map $f\mapsto f-\E f$, where the expectation is with respect to some fixed distribution.   
This generator is sometimes called the \emph{simple} generator~\cite{MOS12}.
The quantum analog of simple generators is
$$\cL(X):=X - \tr(\sigma X) I,$$
for some positive definite density matrix $\sigma$. 
Observe that $\cL$ is primitive, and $\cL^*(X) =X-\tr(X)\sigma$ satisfies the detailed balanced condition with respect to $\sigma$. The quantum Markov semigroup associated to this Lindblad generator is 
\begin{align}\label{eq:def-Phi}
\Phi_t(X)=\e^{-t} X + (1-\e^{-t}) \tr(\sigma X) \II.
\end{align}
In the special case where $\sigma$ is the completely mixed state, $\Phi_t$ and $\Phi_t^*$ coincide and become depolarizing channels. Indeed,~\eqref{eq:def-Phi} is a \emph{generalized depolarizing channel} in the Heisenberg picture.

Having two Lindblad generators $\cL$ and $\cK$ associated to two semigroups $\{\Phi_t:\, t\geq 0\}$ and $\{\Psi_t:\, t\geq 0\}$, respectively, we may consider a new Lindblad generator 
$\cL\otimes \cI+\cI\otimes \cK$. This Lindblad generator generates the semigroup $\{\Phi_t\otimes \Psi_t:\, t\geq 0\}$. Moreover, letting 
\begin{align}\label{eq:def-hat-L}
\widehat \cL_i:= \cI^{\otimes (i-1)}\otimes \cL\otimes \cI^{\otimes (n-i)},
\end{align} 
we have
$$\Phi_t^{\otimes n} = \e^{-t\sum_{i=1}^n \widehat\cL_i}.$$
Note that, if $\cL$ is primitive and reversible with respect to $\sigma$, then $\sum_{i=1}^n \widehat \cL_i$ is also primitive and reversible with respect to $\sigma^{\otimes n}$.

\subsection{Dirichlet form}
We now define the \emph{Dirichlet form}\footnote{Again, our definition of the Dirichlet form is different from that of~\cite{KT13} by a factor of $ p/2$ and a negative sign.} associated to a QMS with generator $\cL$ by
$$\cE_{p, \cL}(X) = \frac{p\hat p}{4}\langle I_{\hat p, p}(X), \cL(X)\rangle_\sigma,$$
where $\hat p$ is the H\"older conjugate of $p$.
Verification of the following properties of the Dirichlet form is easy.

\begin{prop}\label{prop:dirichlet}
\begin{enumerate}
\item[{\rm (i)}] $\cE_{\hat p, \cL}(I_{\hat p, 2}(X)) =\cE_{p, \cL}(I_{p, 2}(X))$ for all $p\in \RR\backslash\{0\} $ and $X\in\cB(\cH)$.
\item [{\rm (ii)}] $\cE_{ p, \cL}(cX)=c^p\cE_{p, \cL}(X)$ for $X\geq 0$ and constant $c\geq 0$.
\item[{\rm (iii)}] $\cE_{2, \cL} (X) = \langle X, \cL(X)\rangle_\sigma$ for all $X> 0$. 
\item[{\rm (iv)}] $\cE_{1, \cL}(X) = \frac{1}{4} \tr\left[\Gamma_\sigma\big(\cL(X)\big)\cdot\big(\log \Gamma_\sigma(X) - \log \sigma \big) \right].$
\end{enumerate}
\end{prop}

The non-negativity of the Dirichlet form is not clear from its definition. Here we prove the non-negativity assuming that the QMS is $p$-contractive. By Proposition~\ref{prop:contraction} we then conclude the non-negativity of $\cE_{p, \cL}(X)$ for $p\notin(-1, 1/2)$.  Later on, based on an stronger assumption than $\sigma$-reversibility, we will prove $\cE_{p, \cL}(X)\geq 0$ for all values of $p$ and $X>0$. 

\begin{prop}\label{prop:dirichlet-positive}
Suppose that $\cL$ generates a QMS that is primitive and $\sigma\in \cD_+(\cH)$ is its unique invariant state. { Let $p\in\RR\ne \{0\}$.} If the QMS is (reverse) $p$-contractive, then $\cE_{p, \cL}(X)\geq 0$ for all $X> 0$. 
\end{prop}

\begin{proof}
Define 
$$g(t) :=\hat p\big\|  \Phi_t(X) \big\|^{p}_{p, \sigma}-\hat p\|X\|^{p}_{p, \sigma} .$$
By assumption of (reverse) $p$-contractivity, for all $t\geq 0$ we have $g(t)\leq 0$. We note that $g(0)=0$. Therefore, $g'(0)\leq 0$. We compute
\begin{align*}
g'(0) & = \frac{\dd}{\dd t}\, \hat p\,\big\|  \Phi_t(X) \big\|^{p}_{p, \sigma}\Big|_{t=0}\\
& = \frac{\dd}{\dd t} \,\hat p\,\tr\Big(  \Gamma_{\sigma}^{\frac 1p}\circ\Phi_t(X)^p\Big)\Big|_{t=0}\\
& = -p\hat p \,\tr\Big( \Gamma_\sigma^{\frac 1p}\circ \cL(X) \cdot \Gamma_\sigma^{\frac 1p} (X)^{p-1}     \Big)\\
& = -p\hat p \,\tr\Big( \cL(X) \cdot \Gamma_\sigma^{\frac 1p}\big(\Gamma_\sigma^{\frac 1p} (X)^{p-1}\big)     \Big)\\
& = -p\hat p \langle I_{\hat p, p}(X), \cL(X)\rangle_\sigma.
\end{align*}
This gives $\cE_{p, \cL}(X)\geq 0$.
\end{proof}

\subsection{Hypercontractivity and logarithmic-Sobolev inequalities}

We showed in Proposition~\ref{prop:contraction} that $\Phi_t$ belonging to a $\sigma$-reversible QMS is contractive, at least for certain values of $p$. That is, 
$\|\Phi_t(X)\|_{p, \sigma}$ is bounded (from above or below depending on whether $p\geq 1$ or $p<1$) by $\|X\|_{p, \sigma}$. On the other hand, By part (b) of Corollary~\ref{corol:ent-positive} bounding $\|\Phi_t(X)\|_{p, \sigma}$ by $\|X\|_{q, \sigma}$ when $1\leq q<p$ or $p<q<1$ is a stronger inequality than contractivity. Such inequalities are called \emph{hypercontractivity inequalities} or \emph{reverse hypercontractivity inequalities} depending on whether $1\leq q<p$ or $p<q<1$ respectively. These inequalities have found a wide range of applications in the literature. 

It is well-known that quantum hypercontractivity inequalities stem from quantum \emph{logarithmic-Sobolev} (log-Sobolev) inequalities. They are essentially equivalent objects, so proving log-Sobolev inequalities gives hypercontractivity ones. The theory of reverse hypercontractivity inequalities have been generalized to the non-commutative case for unital semigroups in~\cite{CMT15}. Here we generalize the theory for general QMS.

Given a primitive Lindblad generator $\cL$ that is reversible with respect to a positive definite density matrix $\sigma$ and $p\in \RR\backslash \{0\}$, a $p$-\emph{log-Sobolev} inequality is an inequality of the form
$$\beta\,\Ent_{p, \sigma}(X)\leq \cE_{p, \cL}(X), \qquad \forall X> 0.$$
The best constant $\beta$ satisfying the above inequality is called the $p$-\emph{log-Sobolev constant} and is denoted by $\alpha_p(\cL)$. That is,
$$\alpha_p(\cL) : = \inf \frac{\cE_{p, \cL}(X)}{\Ent_{p, \sigma}(X)},$$
where the infimum is taken over $X> 0$ with $\Ent_{p, \sigma}(X)\neq 0$.

By the following proposition we can restrict ourselves to log-Sobolev constants for values of $p\in [0, 2]$. 

\begin{prop}\label{prop:LS-const}
$\alpha_p(\cL)=\alpha_{\hat p}(\cL)$ for all Lindblad generators $\cL$.
\end{prop}

\begin{proof}
Identifying $X$ with $I_{p, 2}(Y)$, for some arbitrary $Y> 0$, this is an immediate consequence of part (i) of Proposition~\ref{prop:ent} and part (i) of Proposition~\ref{prop:dirichlet}. 
\end{proof}

We can now state how log-Sobolev inequalities are related to hypercontractivity and reverse hypercontractivity inequalities. The first part of the following theorem is already known~\cite{OZ99, KT13}. 

\begin{theorem}\label{thm:gen-hyper}
Let $\cL$ be a primitive Lindblad generator that is reversible with respect to {a} positive definite density matrix $\sigma$. 
Then the following holds:
\begin{itemize}
\item {\rm (Hypercontractivity)} Suppose that $\beta_2 = \inf_{p\in [1, 2]} \alpha_p(\cL) >0$. Then for $1\leq q\leq p$ and 
\begin{align}\label{eq:def-t-1}
t\geq \frac{1}{4\beta_2}\log\frac{p-1}{q-1},
\end{align}
we have 
$\| \Phi_t(X)\|_{p, \sigma}\leq \|X\|_{q, \sigma}$ for all $X> 0$
\item {\rm (Reverse hypercontractivity)} Suppose that $\beta_1 = \inf_{p\in {(}0, 1]} \alpha_p(\cL) >0$. Then for $p\leq q<1$ and 
\begin{align}\label{eq:def-t-2}
t\geq \frac{1}{4\beta_1}\log\frac{p-1}{q-1},
\end{align}
we have 
$\| \Phi_t(X)\|_{p, \sigma}\geq \|X\|_{q, \sigma}$ for all $X> 0$, {where \Cref{eq:def-t-2} is understood in the limit whenever $p=0$ or $q=0$.}
\end{itemize} 
\end{theorem}

The proof strategy of this theorem is quite standard. Here we present a proof for the sake of completeness.

\begin{proof}

It suffices to prove the theorem when 
$t= \frac{1}{4\beta} \log \frac{p-1}{q-1}$ for $\beta$ being either $\beta_2$ or $\beta_1$ depending on whether we prove the hypercontractivity part or the reverse hypercontractivity part. 
Thus,  fix $q$ and define
$$t(p):= \frac{1}{4\beta} \log \frac{p-1}{q-1}.$$
Define 
$$f(p):=\|\Phi_{t(p)}(X)\|_{p, \sigma}  -\|X\|_{q, \sigma} = \|X_p\|_{p, \sigma} - \|X\|_{q, \sigma},$$
where $X_p:= \Phi_{t(p)}(X)> 0$. To continue the proof we compute the derivative of $f(p)$  using Proposition~\ref{prop:norm-derivative}.
\begin{align*}
f'(p) & = \frac{\dd}{\dd p} \|X_p\|_{p, \sigma} = \frac{1}{p^2}\|X_p\|_{p, \sigma}^{1-p}\cdot\left(\Ent_{p, \sigma}(X_p) + p^2 \tr\Big[ \Gamma_\sigma^{\frac 1p}(Z_p)\cdot \Gamma_\sigma^{\frac 1p}(X_p)^{p-1}   \Big]\right),
\end{align*}
where 
$$Z_p= \frac{\dd}{\dd p} X_p = -t'(p)\cL(X_p)= -\frac{1}{4\beta (p-1)} \cL(X_p).$$
Therefore, 
$$f'(p)  = \frac{1}{p^2}\|X_p\|_{p, \sigma}^{1-p}\cdot\Big(\Ent_{p, \sigma}(X_p) - \frac{1}{\beta} \cE_{p, \cL}(X_p)   \Big).$$

Now suppose that $q\geq 1$ and $\beta\leq \alpha_p(\cL)$ for all $p\in [1, 2]$. Then for $p\geq q$ we have
$$\Ent_{p, \sigma}(X_p)\leq \frac{1}{\alpha_p(\cL)}\cE_{p, \cL}(X_p)\leq \frac{1}{\beta}\cE_{p, \cL}(X_p).$$
As a result, $f'(p)\leq 0$ for all $p\geq q$. Since $f(q)=0$ we conclude that $f(p)\leq 0$ for all $p\geq q$. This gives the hypercontractivity part of the theorem. 

For the reverse hypercontractivity part, assume that $q< 1$ and $\beta\leq \alpha_p(\cL)$ for all $p\in [0, 1]$. Then for $p\leq q$ we have
$$\Ent_{p, \sigma}(X_p)\leq \frac{1}{\alpha_p(\cL)}\cE_{p, \cL}(X_p)\leq \frac{1}{\beta}\cE_{p, \cL}(X_p),$$
where the second inequality holds since $p<1$, so either $p$ or its H\"older conjugate belongs to $[0,1]$. Therefore, $f'(p)\leq 0$ for all $p\leq q< 1$, and since $f(q)=0$, $f(p)\geq 0$ for all $p<q$.

\end{proof}

%******************************************************************************
\section{Quantum Stroock-Varopoulos inequality}

In the previous section we developed the basic tools required to understand quantum hypercontractivity and reverse hypercontractivity inequalities and log-Sobolev inequalities. By Theorem~\ref{thm:gen-hyper} to obtain hypercontractivity and reverse hypercontractivity inequalities we need to find bounds on log-Sobolev constants in ranges $p\in [1, 2]$ or $p\in [0, 1]$. Now the question is how such bounds can be found. 

In the classical (commutative) case, the most relevant $p$-log-Sobolev constants are $\alpha_2(\cL)$ and $\alpha_1(\cL)$. Indeed, $p\mapsto \alpha_p(\cL)$ is a non-increasing function on $p\in [0, 2]$, so in Theorem~\ref{thm:gen-hyper} the parameters  $\beta_1$ and $\beta_2$ can be replaced with $\alpha_1(\cL)$ and $\alpha_2(\cL)$ respectively. This result is proven via 
comparison of the Dirichlet forms, an inequality that is sometimes called the Stroock-Varopoulos inequality. 

In this section we prove a quantum generalization of the Stroock-Varopoulos inequality, and conclude {in Theorem~\ref{thm:gen-hyper} that, for strongly reversible semigroups,} we can take $\beta_p=\alpha_p(\cL)$ for $p=1, 2$. We should point out that a quantum Stroock-Varopoulos inequality  in the special case of $\sigma$ being the completely mixed state is proven in~\cite{CMT15}. Also, a {special case of the Stroock-Varopoulos} inequality (called \emph{strong $L_p$-regularity}) for certain Lindblad generators is proven { in~\cite{OZ99,KT13}}. A strong $L_p$-regularity is also proven in~\cite{BarEID17} which we generalize to a  quantum Stroock-Varopoulos inequality.

The assumption of $\sigma$-reversibility is not enough for us for proving the quantum Stroock-Varopoulos inequality. We indeed need $\cL$ to be self-adjoint with respect to an inner product different from $\langle \cdot, \cdot \rangle_\sigma$ defined above (see Lemma~\ref{lem:reversible}). In the following we first define this new inner product, state some of its properties and then go to our quantum Stroock-Varopoulos inequality.

\subsection{The GNS inner product}

In what follows we use the GNS inner product $\langle \cdot, \cdot\rangle_{1, \sigma}$ on $\cB(\cH)$ that is defined by \cite{CM16}:
\begin{align}\label{eq:inner-product-1-sigma}
\langle X, Y\rangle_{1, \sigma} := \tr(\sigma X^\dagger Y).
\end{align}
We note that this inner product coincides with $\langle X, Y\rangle_\sigma = \tr(\sigma^{1/2}X^\dagger \sigma^{1/2}Y)$ when, e.g., $X$ and $\sigma$ commute. But in general $\langle \cdot, \cdot\rangle_{1, \sigma}$ is different from $\langle \cdot, \cdot\rangle_\sigma$. 

The following lemma was first proven in~\cite{CM16}. We will give a proof here for the sake of completeness. 

\begin{lemma}\label{lem:modular}
Let $\cL$ be a Lindblad generator that is self-adjoint with respect to the inner product
$\langle \cdot, \cdot\rangle_{1, \sigma}$ defined above. Then the followings hold.
\begin{itemize}
\item[{\rm (i)}] $\cL$ commutes with the superoperator $\Delta_{\sigma}:X \mapsto \sigma X\sigma^{-1}.$
\item[{\rm (ii)}] $\cL$ is self-adjoint with respect to the inner product $\langle \cdot, \cdot \rangle_{\sigma}$.
\end{itemize}
\end{lemma}

Based on part (ii) of this lemma (see also Lemma~\ref{lem:reversible}) we say that a Lindblad generator $\cL$ is \emph{strongly $\sigma$-reversible} if it is self-adjoint with respect to the inner product $\langle \cdot, \cdot\rangle_{1, \sigma}$.

\begin{proof}
(i) Using the fact the $\cL(Y)^{\dagger} = \cL(Y^{\dagger})$, for all $X, Y$ we have
\begin{align*}
\langle X, \Delta_\sigma\circ \cL(Y)\rangle_{1, \sigma} & =  \tr(\sigma X^{\dagger}\sigma \cL(Y)\sigma^{-1}) \\
& = \tr( X^{\dagger}\sigma \cL(Y)) \\
& = \langle \cL(Y)^{\dagger}, X^{\dagger}\rangle_{1, \sigma}\\
&= \langle \cL(Y^\dagger), X^{\dagger}\rangle_{1, \sigma}\\
&= \langle Y^\dagger, \cL(X^{\dagger})\rangle_{1, \sigma}\\
& = \tr(\sigma Y \cL(X)^{\dagger})\\
&= \tr(\Delta_{\sigma}(Y) \sigma \cL(X)^\dagger)\\
& = \langle \cL(X), \Delta_\sigma(Y)\rangle_{1, \sigma}\\
& = \langle X, \cL\circ\Delta_\sigma(Y)\rangle_{1, \sigma}.
\end{align*}
This gives $\Delta_\sigma\circ \cL = \cL\circ\Delta_\sigma$.

(ii) Follows easily from (i) and the fact that
$$\langle X, Y\rangle_{\sigma} = \langle Y^{\dagger}, \Delta_{\sigma}^{1/2}(X^{\dagger})\rangle_{1, \sigma}.$$
\end{proof}

The following lemma is indeed a consequence of Theorem~3.1 of~\cite{CM16}. Here we prefer to present a direct proof.

\begin{lemma}\label{lem:choi}
Let $\cL$ be a strongly $\sigma$-reversible Lindblad generator. Then for every $t\geq 0$ there are operators $R_k\in \cB(\cH)$ and $\omega_k> 0$ such that $\Delta_\sigma (R_k) = {\omega_k} R_k$,
\begin{align}\label{eq:phi-Kraus}
\Phi_t(X) = \sum_k R_k XR_k^\dagger,
\end{align}
and $\sum_k R_k R_k^{\dagger}=I$.
\end{lemma}

\begin{proof}
By Lemma~\ref{lem:modular} the Lindblad generator $\cL$ and then $\Phi_t=\e^{-t\cL}$ commute with~$\Delta_\sigma$, i.e., 
\begin{align}\label{eq:phi-delta-commute}
\Phi_t\circ \Delta_\sigma = \Delta_\sigma\circ \Phi_t.
\end{align}
Fix an orthonormal basis $\{\ket i\}_{i=1}^d$ for the underlying Hilbert space $\cH=\cH_A$ and define
$$\ket{\Upsilon} := \sum_{i=1}^d \ket i_A\ket i_B \in \cH_{AB},$$
where $\cH_B$ is isomorphic to $\cH_A$.
It is not hard to verify that for any matrix $M$ we have 
\begin{align}\label{eq:M-transpose}
(M_A\otimes I_B)\ket{\Upsilon} = \II_A\otimes M_B^T\ket{\Upsilon},
\end{align}
where the transpose is with respect to the basis $\{\ket i\}_{i=1}^d$.

The Choi-Jamiolkowski representation of $\Phi_t$ is 
$$J_{AB} : = (\Phi_t \otimes \cI_B)(\ket{\Upsilon}\bra{\Upsilon}).$$
Then using~\eqref{eq:M-transpose} it is not hard to verify that~\eqref{eq:phi-delta-commute} translates to
$$(\sigma_A^{-1}\otimes \sigma_B^T) J_{AB} = J_{AB}( \sigma_A^{-1}\otimes \sigma_B^T).$$ 
That is, $J_{AB}$ and $\sigma_A^{-1}\otimes \sigma_B^{T}$ commute. On the other hand, $J_{AB}$ is positive semidefinite since it is the  Choi-Jamiolkowski representation of a completely positive map. Therefore, $J_{AB}$ and $\sigma_A^{-1}\otimes \sigma_B^{T}$ can be simultaneously diagonalized in an orthonormal basis, i.e., there exists an orthonormal basis $\{\ket{v_k}\}_{k=1}^{d^2}$ of $\cH_{AB}$ such that 
\begin{align}
J_{AB}\ket{v_k}&= \lambda_k \ket{v_k}\label{eq:eigen-J}\\
\sigma_A^{-1}\otimes \sigma_B^{T}\ket{v_k}&= \omega_k^{-1}\ket{v_k},\label{eq:eigen-sigma}
\end{align}
where $\lambda_k\ge 0,\, \omega_k> 0$. Define the operator $V_k$ by
$$(V_k\otimes I_B)\ket{\Upsilon} = \ket{v_k}.$$ 
Then again using~\eqref{eq:M-transpose}, equation~\eqref{eq:eigen-sigma} translates to 
\begin{align*}%\label{eq:s-v-k}
\sigma^{-1}V_k \sigma = \omega_k^{-1} V_k.
\end{align*}
Moreover, equation~\eqref{eq:eigen-J} means that 
$$(\Phi_t\otimes \cI_B)(\ket{\Upsilon}\bra{\Upsilon})=J_{AB}=\sum_k \lambda_k \ket{v_k}\bra{v_k} = \sum_k \lambda_k (V_k\otimes I_B) \ket{\Upsilon}\bra{\Upsilon} (V_k^{\dagger}\otimes I_B), $$  
which gives
$$\Phi_t(X) :=\sum_k \lambda_k V_kXV_k^{\dagger}.$$
Then letting $R_k:= \sqrt{\lambda_k} V_k$ we have $\sigma R_k= \omega_kR_k \sigma$ and~\eqref{eq:phi-Kraus} holds. The other equation comes from $\Phi_t(\II)=\II$.
\end{proof}

\subsection{Comparison of the Dirichlet forms}

We can now state the main result of this section. 

\begin{theorem}[Quantum Stroock-Varopoulos inequality]\label{thm:QSVineq}
Let $\cL$ be a strongly $\sigma$-reversible Lindblad generator, which means that it is self-adjoint with respect to the inner product
$\langle \cdot, \cdot\rangle_{1, \sigma}$ defined in~\eqref{eq:inner-product-1-sigma}.
Then for all $X> 0$ we have
$$\cE_{p, \cL}\big(I_{p, 2}(X)\big)\geq \cE_{q, \cL}\big(I_{q, 2}(X)\big), \qquad 0{<} p\leq q\leq 2.$$
\end{theorem}

{
\begin{remark}
	As mentioned above, special cases of \Cref{thm:QSVineq} were already investigated in the literature. In the case that $\sigma$ is the maximally mixed state, this was done in \cite{CMT15}. The inequality was also recently extended to the GNS-symmetric setting for the range of parameters $p\ge 1$ and $q=2$ in \cite{BarEID17}.
	\end{remark}
}
We have two proofs for this theorem. The first one, that we present here, is based on ideas in~\cite{OZ99, KT13}. The second one, that is moved to Appendix~\ref{app:qSV}, is based on ideas in~\cite{BarEID17}. We present both the proofs in this paper since they are different in nature and whose ideas can be useful elsewhere.

\begin{proof}[First proof of Theorem~\ref{thm:QSVineq}]
For any $t\geq 0$ define the function $h_t:{[}0,\infty)\to\RR$ by
$$h_t(s):= \big\langle I_{2/(2-s), 2}(X) , \Phi_t \circ I_{2/s, 2}(X)\big\rangle_\sigma$$
{for $s\in(0,\infty)\backslash\{2\}$, and $h_t(0)=h_t(2)=\tr(\Gamma_\sigma^{1/2}(X)^2)$.} Since by part (ii) of Lemma~\ref{lem:modular}, $\Phi_t=\e^{-t\cL}$ is self-adjoint with respect to the inner product $\langle \cdot, \cdot\rangle_\sigma$, we have $h_t(2-s)=h_t(s)$ and $h_t$ is symmetric about $s=1$. Moreover, exploring the definition of $h_t(s)$ we find that $s\mapsto h_t(s)$ is analytic with a convergent Taylor series at $s=1$. 
Then, by the symmetry around $s=1$, all the the odd-order derivatives of $h_t$ at $s=1$ vanish, and we have
\begin{align}\label{eq:hts}
h_t(s) = h_t(1) + \sum_{j=1}^\infty \frac{c_j}{(2j)!}  (s-1)^{2j},
\end{align}
where 
\begin{align}\label{cj}
c_j = \frac{\dd^{2j}}{\dd s^{2j}} h_t(s)\Big|_{s=1}\,.
\end{align}
{Note that the above series expansion is convergence by analyticity of $s\mapsto h_t(s)$.}We claim that all the even-order derivatives of $h_t$ at $s=1$ are non-negative, i.e., $c_j\geq 0$. We use Lemma~\ref{lem:choi} to verify this. Let $R_k$'s be operators such that
\begin{align}\label{eq:R-k-commute}
\sigma R_k\sigma^{-1} = \omega_kR_k,
\end{align}
with $\omega_k> 0$ and~\eqref{eq:phi-Kraus} holds. Then letting $Y:= \Gamma_\sigma^{1/2}(X)$ and using~\eqref{eq:R-k-commute} we compute
\begin{align*}
h_t(s) & = \tr\Big[ \Gamma_\sigma^{\frac s 2}(Y^{2-s})\cdot \Phi_t\big(\Gamma_\sigma^{-\frac s 2}(Y^s)\big)   \Big] \\
& = \sum_k \tr\Big[ Y^{2-s} \sigma^{\frac s4} R_k \sigma^{-\frac s4} Y^s \sigma^{-\frac s4} R_k^{\dagger} \sigma^{\frac s4}       \Big]\\
& =\sum_k \omega_k^{\frac{s}{2}}\tr\Big[ Y^{2-s} R_k  Y^s  R_k^{\dagger}        \Big].
\end{align*} 
Now diagonalizing $Y$ in its eigenbasis: $Y=\sum_\ell \mu_\ell\ket { \ell}\bra { \ell}$, we find that
$$h_t(s) = \sum_{k, \ell, \ell'} \mu_\ell^2 \, \big|  \bra{\ell} R_k\ket{\ell'} \big|^2\Big(\frac{\sqrt{\omega_k}\,\mu_{\ell'}}{\mu_\ell}\Big)^s .$$
Therefore, $h_t(s)$ is a sum of exponential functions with positive coefficients. 
From this expression it is clear that $c_j$'s {as defined in \Cref{cj}} are all non-negative.

{For $s\in(0,\infty)\backslash\{2\}$, }let us define 
$$g_t(s):= \frac{h_t(s)-h_t(0)}{(s-1)^2-1} = \sum_{j=1}^\infty \frac{c_j}{(2j)!}\left(\sum_{i=0}^{j-1} (s-1)^{2i} \right),$$
{and extend the function $g_t$ by continuity on $[0,\infty)$, since $h_t$ is differentiable at $0$ and at $2$.} From this expression it is clear that $g_t(s)$ is non-decreasing on $[1, +\infty)$.
Therefore, $\lim_{t\rightarrow 0^+} g_t(s)/t$ is non-decreasing on $[1, +\infty)$. On the other hand, we have $h_t(0) =\tr(Y^2) = h_0(s)$. We thus can compute
\begin{align*}
\lim_{t\rightarrow 0^+} \frac{g_t(s)}{t} & =\frac{1}{(s-1)^2-1} \lim_{t\rightarrow 0^+}  \frac{h_t(s)-h_t(0)}{t} \\
& =\frac{1}{(s-1)^2-1} \lim_{t\rightarrow 0^+}  \frac{h_t(s)-h_0(s)}{t} \\
& =\frac{1}{(s-1)^2-1} \frac{\partial}{\partial t} h_t(s)\Big|_{t=0} \\
& =-\frac{1}{(s-1)^2-1}  \big\langle I_{2/(2-s), 2}(X) , \cL \circ I_{2/s, 2}(X)\big\rangle_\sigma.
\end{align*}
Therefore 
$$s\mapsto -\frac{1}{(s-1)^2-1}  \big\langle I_{2/(2-s), 2}(X) , \cL \circ I_{2/s, 2}(X)\big\rangle_\sigma,$$
is non-decreasing on $[1, +\infty)$. Now the desired result follows once we identify $2/s$ with $p$ (and $2/(2-s)$ with $\hat p$, its H\"older conjugate).
\end{proof}

%%%
Here are some important consequences of the above theorem.

\begin{corollary}\label{cor:NN-Dirichlet}
Let $\cL$ be a strongly $\sigma$-reversible Lindblad generator. Then the followings hold:
\begin{itemize}
\item[{\rm (i)}] For all $p\in \mathbb R\backslash \{0\}$ and $X> 0$ we have
$$\cE_{p, \cL}(X)\geq 0.$$
\item[{\rm (ii)}] The associated QMS is $p$-contractive for all $p$. 
\end{itemize}
\end{corollary}

\begin{remark}
As mentioned before, {the fact that $p$-Dirichlet forms are positive for $p\in (-\infty,- 1]\cup [+1/2,\infty)$ is a simple consequence of contraction of non-commutative weighted $L_p$-norms (or equivalently of the data processing inequality for sandwiched $p$-R\'enyi divergences), which follows by invariance of the state $\sigma$ and interpolation of these spaces (see \cite{OZ99}). The case $p\in (-1,+1/2)$ is much more subtle, since it is known that the data processing inequality does not hold in general in this parameter range, as opposed to the classical case. More precisely, $p$-contractivity of $\Phi_t$ implies that the sandwiched $p$-R\'enyi divergence is monotone under $\Phi_t$ \cite{FL13, Beigi13, MDSFT13}. Therefore, when $\Phi_t$ comes from a QMS satisfying the above strong reversibility condition, sandwiched $p$-R\'enyi divergences are monotone under $\Phi_t$ not only for $p\geq 1/2$ but for \emph{all} values of $p$. }
\end{remark}

\begin{proof}
(i) By Theorem~\ref{thm:QSVineq} (and part (i) of Proposition~\ref{prop:dirichlet}) for every $p\ne 0$ we have
$$\cE_{p, \cL}(I_{p, 2}(X))\geq \cE_{2, \cL}(X).$$ 
{Indeed, for $p\in(0,2]$, the inequality holds by \Cref{thm:QSVineq}, and for $p\notin[0,2]$, further use \Cref{prop:dirichlet}(i) to conclude}. On the other hand, {since we have self-adjointness of the semigroup with respect to $\langle.,.\rangle_\sigma$, its generator has positive spectrum, so that we have $\cE_{2, \cL}(X)\geq 0$}. Therefore, $\cE_{p, \cL}(I_{p, 2}(X))\geq 0$. 

\medskip
\noindent (ii) Define $g(t)$ as in the proof of Proposition~\ref{prop:dirichlet-positive}. By part (i) we have $g'(t)\leq 0$ for all $t\geq 0$ and $g(0)=0$. Therefore, $g(t)\geq 0$ for all $t\geq 0$. This gives $p$-contractivity.  
\end{proof}

The following corollary is an immediate consequence of the quantum Stroock-Varopoulos inequality as well as part (i) of  Proposition~\ref{prop:ent}

\begin{corollary}\label{eq:mon-LS-constant}
Let $\cL$ be a strongly $\sigma$-reversible Lindblad generator. Then 
$p\mapsto \alpha_p(\cL)$ is non-increasing on $[0, 2]$, where $\alpha_0(\cL)$ is defined as the limit $p\to 0$.
\end{corollary}

Now we can state an improvement over Theorem~\ref{thm:gen-hyper}.

\begin{corollary}\label{cor:gen-hyper-1}
Let $\cL$ be a strongly $\sigma$-reversible Lindblad generator. Then the following holds:
\begin{itemize}
\item {\rm (Hypercontractivity)} For $1\leq q\leq p$ and 
\begin{align}\label{eq:def-t-1'}
t\geq \frac{1}{4\alpha_2(\cL)}\log\frac{p-1}{q-1},
\end{align}
we have 
$\| \Phi_t(X)\|_{p, \sigma}\leq \|X\|_{q, \sigma}$ for all $X\geq 0$
\item {\rm (Reverse hypercontractivity)} For $p\leq q<1$ and 
\begin{align}\label{eq:def-t-2'}
t\geq \frac{1}{4\alpha_1(\cL)}\log\frac{p-1}{q-1},
\end{align}
we have 
$\| \Phi_t(X)\|_{p, \sigma}\geq \|X\|_{q, \sigma}$ for all $X> 0$.
\end{itemize} 
\end{corollary}

\begin{remark}
	{\Cref{eq:def-t-1'} was already known to be implied by the strong $L_p$-regularity defined by \cite{OZ99}. This condition, which is a special case of the Stroock-Varopoulos inequality, was recently shown in \cite{BarEID17}. }
	\end{remark}
Before ending this section, we state a result that will play an important role in \Cref{sec6}.
\begin{lemma}\label{RHolderHC}
	Let $\{\Phi_t:\, t\ge 0\}$ be a  a primitive QMS that is strongly $\sigma$-reversible. Let $X,Y>0$ and $-\infty\le q, p< 1$. Then, for any $t\ge 0$ such that $(1-p)(1-q)\ge\e^{-4\alpha_1(\cL) t}$ we have
	\begin{align*}
		\langle X,\Phi_t(Y)\rangle_\sigma \ge \|X\|_{p,\sigma}\|Y\|_{q,\sigma}
	\end{align*}
	
\end{lemma}

\begin{proof}
	The result follows by a direct application of \Cref{RHolder} together with the reverse hypercontractivity inequality in \Cref{cor:gen-hyper-1}.
	
\end{proof}

%**************************************************

\section{Tensorization}

Our goal in this section is to prove hypercontractivity (or reverse hypercontractivity) inequalities of the form $\|\Phi_t^{\otimes n}(X)\|_{p, \sigma^{\otimes n}} \leq \|X\|_{q, \sigma^{\otimes n}}$ (or $\|\Phi_t^{\otimes n}(X)\|_{p, \sigma^{\otimes n}} \geq  \|X\|_{q, \sigma^{\otimes n}}$) for certain ranges of $t, p, q$ that are \emph{independent of $n$}. Indeed, so far we have a theory of using log-Sobolev inequalities to prove such inequalities when $n=1$, but in some applications, e.g., those we present later in this paper, we need such inequalities for arbitrary $n$.  
We need some notations to state the problem more precisely. 

For a Lindblad generator $\cL$ we define
\begin{align}\label{eq:hat-L-i}
\widehat \cL_i:= \cI^{\otimes (i-1)}\otimes \cL\otimes \cI^{\otimes (n-i)},
\end{align}
as an operator acting on $\cB(\cH^{\otimes n})$. We also 
let
\begin{align}\label{eq:K-n}
 \cK_n:= \sum_{i=1}^n \widehat \cL_i.
\end{align}
Observe that if $\cL$ is (strongly) $\sigma$-reversible, then $\cK_n$ is (strongly) reversible with respect to $\sigma^{\otimes n}$. Moreover, $\widehat \cL_i$'s commute with each other and 
$$\e^{-t\cK_n} = \Phi_t^{\otimes n}.$$
That is, $\cK_n$ is a (strongly) $\sigma^{\otimes n}$-reversible Lindblad generator which generates the quantum Markov semigroup $\big\{\Phi_t^{\otimes n}:\, t\geq 0\big\}$. Now we can ask how the (reverse) hypercontractivity inequalities associated to $\Phi_t$ are related to those for $\Phi_t^{\otimes n}$. Equivalently, what is the relation between the log-Sobolev constants $\alpha_p(\cL)$ to $\alpha_p(\cK_n)$? In the commutative (classical) case the answer is easy;  $\alpha_p(\cK_n)$ equals $\alpha_p(\cL)$ for all $n$, and having a (reverse) hypercontractivity inequality for $\Phi_t$ immediately gives one for $\Phi_t^{\otimes n}$.  This is because in the classical case operator norms are multiplicative, or because the entropy function satisfies a certain subadditivity property (see e.g.,~\cite{MOS12}).
The aforementioned  property that, in the classical case, $\alpha_p(\cK_n)$ is independent of $n$, is usually called the \emph{tensorization} property. 

Tensorization property of log-Sobolev constants of quantum Lindblad generators, unlike its classical counterpart, is highly non-trivial. Thus proving (reverse) hypercontractivity inequalities that are independent of $n$ is a difficult problem in the non-commutative case. There are some attempts in this direction. Montanaro and Osborne in~\cite{MO10} 
proved such hypercontractivity inequalities for the qubit depolarizing channel (see also~\cite{KT13}). King~\cite{King14} generalized this result for all unital qubit QMS. 
Cubitt \emph{et al.}~developed the theory of quantum reverse hypercontractivity inequalities in the unital case in~\cite{CMT15} and proved some tensorization-type results. 
Also, Cubitt \emph{et al.}~\cite{TPK14} developed some techniques for proving bounds on log-Sobolev constants $\alpha_p(\cK_n)$ that are independent of $n$. Beigi and King~\cite{BK16} took the path of developing the theory of log-Sobolev inequalities not for the usual $q\to p$ norm, but for the completely bounded norm. The point is that completely bounded norms are automatically multiplicative~\cite{DJKR16}, so there is no problem of tensorization for the associated log-Sobolev constants. However, the existence of a complete version of the LSI constant was disproved in \cite{bardet2018hypercontractivity}.

In this section we prove two tensorization-type results, one for $1$-log-Sobolev constants which will be used for reverse hypercontractivity inequalities, and the other for $2$-log-Sobolev constants which would be useful for hypercontractivity inequalities. 

\begin{theorem}\label{thm:tensorization-alpha-1}
Let $\sigma_1, \dots, \sigma_n$ be arbitrary positive definite density matrices. Let $\cL_i(X) = X-\tr(\sigma_i X) \II$ be the simple generator associated to the state $\sigma_i$. Let
$$\widehat \cL_i:= \cI^{\otimes (i-1)}\otimes \cL_i\otimes \cI^{\otimes (n-i)},
$$
and define $\cK_n$ by~\eqref{eq:K-n}. Then we have
$\alpha_1(\cK_n) \geq \frac 14$, independently of $n$.
\end{theorem}
\begin{remark}
	{Observe that \Cref{thm:tensorization-alpha-1} does not show the tensorization of $\alpha_1$ for the depolarizing semigroup, but only proves a positive lower bound independent of $n$. Hence, the tensorization of $\alpha_1$ is still an open problem.}
	\end{remark}

Letting $\sigma_i$'s to be equal in the above theorem, we obtain the promised tensorization-type result for the $1$-log-Sobolev constant.\footnote{Note that this result was independently obtained recently in \cite{capel2018quantum} by introducing the notion of a conditional log-Sobolev constant and finding a uniform lower bound on the latter. Moreover, a special case of the above theorem corresponding to $\sigma$ being the completely mixed state was proved in~\cite{MFW16}.
}

\begin{proof}
We need to show that for all $X_{A^n}\in \cP_+(\cH_{A^n})$ we have
\begin{align*}
\frac{1}{4} \Ent_{1, \sigma_{A^n}}(X_{A^n})\leq \cE_{1, \cK_n}(X_{A^n}),
\end{align*}
where $\sigma_{A_i} =\sigma_i$ and 
$$\sigma_{A^n} = \sigma_1\otimes \cdots \otimes \sigma_n.$$
Using parts (ii) of Proposition~\ref{prop:ent} and Proposition~\ref{prop:dirichlet}, without loss of generality we can assume that $X_{A^n}= \Gamma_{\sigma_{A^n}}^{-1}(\rho_{A^n})$ where $\rho_{A^n}\in \cD_+(\cH_{A^n})$ is a density matrix.  Then, using parts (iv) of Proposition~\ref{prop:ent} and Proposition~\ref{prop:dirichlet}, we need to show that
\begin{align}\label{eq:dd1}
D(\rho_{A^n}\|\sigma_{A^n})\leq \sum_{i=1}^n \tr\Big[ \Gamma_{\sigma_{A^n}}\circ \widehat \cL_i\circ \Gamma_{\sigma_{A^n}}^{-1} (\rho_{A^n})\cdot \big( \log \rho_{A^{n}} - \log(\sigma_{A^n})     \big)\Big].
\end{align}
Observe that 
\begin{align*}
\Gamma_{\sigma_{A^n}}\circ \widehat \cL_i\circ \Gamma_{\sigma_{A^n}}^{-1} & =  \cI^{\otimes (i-1)}\otimes \big(\Gamma_{\sigma_i}\circ\cL\circ\Gamma_{\sigma_i}^{-1}\big)\otimes \cI^{\otimes (n-i)}
 =  \cI^{\otimes (i-1)}\otimes \cL^*_i\otimes \cI^{\otimes (n-i)},
\end{align*}
with $\cL^*_i (Y) = Y- \tr(Y)\sigma_i$. Therefore, 
$$\Gamma_{\sigma_{A^n}}\circ \widehat \cL_i\circ \Gamma_{\sigma_{A^n}}^{-1} (\rho_{A^n}) = \rho_{A^n} - \rho_{A^{\sim i}}\otimes \sigma_{A_i},$$
where $A^{\sim i} = (A_1, \dots, A_{i-1}, A_{i+1}, \dots, A_n)$ and $\rho_{A^{\sim i}} = \tr_{A_i}(\rho_{A^n})$ is the partial trace of $\rho_{A^n}$ with respect to the $i$-th subsystem. 
Therefore,~\eqref{eq:dd1} is equivalent to 
\begin{align*}
D(\rho_{A^n}\|\sigma_{A^n}) &\leq  \sum_{i=1}^n\tr\Big[ \big(\rho_{A^n} - \rho_{A^{\sim i}}\otimes \sigma_{A_i} \big)\cdot \big( \log \rho_{A^{ n}} - \log(\sigma_{A^n})     \big)\Big]\\
&= \sum_{i=1}^n\Big[ D(\rho_{A^n}\|\sigma_{A^n}) + D(\rho_{A^{\sim i}}\otimes \sigma_{A_i}\| \rho_{A^n}) - D(\rho_{A^{\sim i}}\otimes \sigma_{A_i}\| \sigma_{A^n})\Big].
\end{align*}
Now since $D(\rho_{A^{\sim i}}\otimes \sigma_{A_i}\| \rho_{A^n})\geq 0$, it suffices to show that
\begin{align}
D(\rho_{A^n}\|\sigma_{A^n}) &\leq  \sum_{i=1}^n\Big[ D(\rho_{A^n}\|\sigma_{A^n})  - D(\rho_{A^{\sim i}}\otimes \sigma_{A_i}\| \sigma_{A^n})\Big].
\label{eq:dd2}
\end{align}
We note that $D(\xi_B\| \tau_B) = -H(B)_\xi -\tr(\xi\log \tau)$ where $H(B)_\xi = -\tr(\xi\log \xi)$ is the von Neumann entropy. Moreover, $\log (\xi\otimes \tau) = \log \xi\otimes I + I\otimes \log \tau$. Therefore,~\eqref{eq:dd2} is equivalent to  
\begin{align*}
-H(A^n)_{\rho} - \sum_{i=1}^n \tr(\rho_{A_i}\log \sigma_i) &\leq \sum_{i=1}^n\Big[  -H(A^n)_\rho -\sum_{j=1}^n \tr(\rho_{A_j}\log \sigma_j)+ H(A^{\sim i})_\rho + \sum_{j\neq i} \tr(\rho_{A_j}\log \sigma_j) \Big]\\
& =\sum_{i=1}^n\big[ - H(A^n)_\rho - \tr(\rho_{A_i}\log \sigma_i)+ H(A^{\sim i})_\rho  \big]\\
 & =\sum_{i=1}^n\big[ - H(A_i| A^{\sim i})_\rho - \tr(\rho_{A_i}\log \sigma_i)  \big].
\end{align*}
This is equivalent to 
\begin{align*}
H(A^n)_{\rho} &\geq\sum_{i=1}^n  H(A_i| A^{\sim i})_\rho, 
\end{align*}
which is an immediate consequence of the data processing inequality (i.e., $H(B|C)_\xi\geq H(B|CD)_\xi$) once we use the chain rule
$$H(A^n)_{\rho} = H(A_1)_{\rho}+\sum_{i=2}^n H(A_i| A_1, \dots, A_{i-1})_{\rho}.$$
This conclude the proof.

\end{proof}
\begin{remark}
	A similar proof was recently and independently obtained  in \cite{capel2018quantum}. {Moreover, the proof uses similar ideas to the proof of the tensorization property of the variant of $\alpha_2$ for the completely bounded norm in \cite{BK16}.}
\end{remark}

We can now use Corollary~\ref{cor:gen-hyper-1} and the fact that the simple generator is strongly reversible to conclude the following. 

\begin{corollary}\label{cor:reverse-HC-depolarizing}
Let $\sigma_1, \dots, \sigma_n$ be arbitrary positive definite density matrices. Let $\cL_i(X) = X-\tr(\sigma_i X) \II$ be the simple generator associated to the generalized depolarizing channel $\Phi_{t, i}(X)=\e^{-t} X + (1-\e^{-t}) \tr(\sigma_i X) \II$. 
Define $\sigma^{(n)} = \sigma_1\otimes \cdots \otimes \sigma_n$ and $\Phi_t^{(n)} = \Phi_{t, 1}\otimes \cdots \otimes \Phi_{t, n}$.
Then for $p\leq q<1$ and  $t\geq \log\frac{p-1}{q-1}$
 we have 
$$\big\| \Phi_t^{(n)}(X)\big\|_{p, \sigma^{(n)}}\geq \|X\|_{q, \sigma^{(n)}}, \qquad \forall n\geq 1,$$ 
where $X\in \cP_+(\cH^{\otimes n})$ is arbitrary. 
\end{corollary}

We now state the second tensorization result which is about the $2$-log-Sobolev constant.

\begin{theorem}\label{thm:tensorization-alpha-2}
Let $\dim \cH=2$ and $\cL(X) = X-\tr(\sigma X) \II$ for some positive definite density matrix $\sigma\in \cD_+(\cH)$. Then we have 
$$\alpha_2(\cK_n)= \alpha_2(\cL), \qquad \forall n,$$
where $\cK_n$ is defined in~\eqref{eq:K-n}.
\end{theorem}
Our main tool to prove this theorem is the following entropic inequality that is of independent interest and can be useful elsewhere. 

\begin{lemma}\label{lem:entropy-inequality}
Let $\cH$ and $\cH'$ be Hilbert spaces with $\dim \cH=2$. Let $X\in \cP(\cH\otimes \cH')$ be a positive semidefinite matrix with the block form
\begin{align}\label{eq:X-block}
X=\begin{pmatrix}
A & C\\
C^{\dagger} & B
\end{pmatrix},
\end{align}
where $A, B, C\in \cB(\cH')$. For a density matrix $\rho\in \cD_+(\cH')$, the matrix $M$ defined as
\begin{align}\label{eq:M-block}
M= \begin{pmatrix}
\|A\|_{2, \rho} & \|C\|_{2, \rho}\\
\|C^{\dagger}\|_{2, \rho} & \|B\|_{2, \rho}
\end{pmatrix}
\end{align}
is positive semidefinite. Moreover, let $\sigma\in \cD_+(\cH)$ be a density matrix of the form 
\begin{align}\label{eq:sigma-diag}
\sigma = \begin{pmatrix}
\theta & 0\\
0 & 1-\theta
\end{pmatrix},
\end{align}
where $\theta\in (0,1)$. Then we have 
\begin{align}
\Ent_{2, \sigma\otimes \rho}(X)\leq &~ \Ent_{2, \sigma}(M) + \theta \Ent_{2, \rho}(A) +(1-\theta)\Ent_{2, \rho}(B)\nonumber\\
& ~+ \sqrt{\theta(1-\theta)}\,\Ent_{2, \rho}(I_{2, 2}(C)) + \sqrt{\theta(1-\theta)}\,\Ent_{2, \rho}(I_{2, 2}(C^\dagger))\,,
\end{align}
{where the map $I_{2,2}$ is defined with respect to the state $\rho$.}
\end{lemma}

\begin{proof}
For any $p\geq 2$ define 
$$M_p := \begin{pmatrix}
\|A\|_{p, \rho} & \|C\|_{p, \rho}\\
\|C^{\dagger}\|_{p, \rho} & \|B\|_{p, \rho}
\end{pmatrix},$$
so that $M_2=M$. Since $X\geq 0$, both $A$ and $B$ are positive semidefinite. Moreover, we have 
$$\Gamma_{\II\otimes \rho}^{\frac 1p}(X)=  \begin{pmatrix}
\Gamma_\rho^{\frac 1p}(A) & \Gamma_\rho^{\frac 1p}(C)\\
\Gamma_\rho^{\frac 1p}(C^{\dagger}) & \Gamma_\rho^{\frac 1p}(B)
\end{pmatrix}\geq 0.$$
As a result, according to~Theorem IX.5.9 of \cite{Bhatia15} there exists a \emph{contraction} $R\in \cB(\cH')$ such that $\Gamma_\rho^{\frac 1p}(C) = \big(\Gamma_\rho^{\frac 1p}(A)\big)^{\frac 12} R \big(\Gamma_\rho^{\frac 1p}(B)\big)^{\frac 12}$. Therefore, by H\"older's inequality we have
\begin{align*}
\big\|\Gamma_\rho^{\frac 1p}(C)\big\|_{p} & =\big\|\big(\Gamma_\rho^{\frac 1p}(A)\big)^{\frac 12} R \big(\Gamma_\rho^{\frac 1p}(B)\big)^{\frac 12}\big\|_{p} \\
& \leq \big\|\big(\Gamma_\rho^{\frac 1p}(A)\big)^{\frac 12}\big\|_{ 2p} \cdot \|R\|_{\infty}\cdot\big\|\big(\Gamma_\rho^{\frac 1p}(B)\big)^{\frac 12}\big\|_{ 2p}\\
& \leq \big\|\big(\Gamma_\rho^{\frac 1p}(A)\big)^{\frac 12}\big\|_{ 2p} \cdot \big\|\big(\Gamma_\rho^{\frac 1p}(B)\big)^{\frac 12}\big\|_{ 2p}\\
& = \big\|\Gamma_\rho^{\frac 1p}(A)\big\|^{\frac 12}_{ p} \cdot \big\|\Gamma_\rho^{\frac 1p}(B)\big\|^{\frac12}_{ p}.
\end{align*}
Then using $\|Y\|_{p, \rho} = \|\Gamma_\rho^{1/p}(Y)\|_{p}$, we find that 
$$\|C\|_{p, \rho} \leq \|A\|_{p, \rho}^{\frac 12}\cdot\|B\|_{p, \rho}^{\frac 12},$$
and hence $M_p\geq 0$. In particular, $M_2=M\geq 0$ and $\Ent_{2, \rho}(M)$ is well-defined. 

Define 
$\psi(p):= \|M_p\|_{p, \sigma} - \|X\|_{p, \sigma\otimes \rho}$. It is shown by King~\cite{King03} that $\psi(p)\geq 0$ for all $p\geq 2$. Indeed, this inequality is proven in~\cite{King03} in the special case where $\sigma$ and $\rho$ are the identity operators on the relevant spaces. Nevertheless, we have 
$$\|X\|_{p, \sigma \otimes \rho} = \left\|
\begin{pmatrix}
\theta^{\frac 1p}\Gamma_{\rho}^{\frac 1p}(A) & \big(\theta(1-\theta)\big)^{\frac 1{2p}}\Gamma_{\rho}^{\frac 1p}(C)\\
\big(\theta(1-\theta)\big)^{\frac 1{2p}}\Gamma_{\rho}^{\frac 1p}(C^{\dagger}) & (1-\theta)^{\frac 1p}\Gamma_{\rho}^{\frac 1p}(B)
\end{pmatrix}\right\|_{p},$$
and
$$\|M_p\|_{p, \sigma} = \left\|
\begin{pmatrix}
\theta^{\frac 1p}\|\big\|\Gamma_{\rho}^{\frac 1p}(A)\big\|_{p} & \big(\theta(1-\theta)\big)^{\frac 1{2p}}\big\|\Gamma_{\rho}^{\frac 1p}(C)\big\|_{p}\\
\big(\theta(1-\theta)\big)^{\frac 1{2p}}\big\|\Gamma_{\rho}^{\frac 1p}(C^{\dagger})\big\|_{p} & (1-\theta)^{\frac 1p}\big\|\Gamma_{\rho}^{\frac 1p}(B)\big\|_{p}
\end{pmatrix}\right\|_{p},$$
Thus, King's result holds for arbitrary $\rho$ and diagonal $\sigma$ as well, and we have $\psi(p)\geq 0$ for all $p\geq 2$. On the other hand,  a straightforward computation verifies that $\psi(2)=0$. This means that 
$\psi'(2)\geq 0$, i.e.,
$$\frac{\dd}{\dd p} \big(\|M_p\|_{p, \sigma} - \|X\|_{p, \sigma\otimes \rho}\big) \bigg|_{p=2}\geq 0.$$
The derivatives can be computed using Proposition~\ref{prop:norm-derivative}.  We have 
\begin{align}\label{eq:der-X}
\frac{\dd}{\dd p}\|X\|_{p,\sigma\otimes \rho} \bigg|_{p=2}= \frac{1}{4}\|X\|_{2, \sigma\otimes \rho}^{-1}\cdot  \Ent_{2, \sigma\otimes \rho}(X),
\end{align}
and
$$\frac{\dd}{\dd p}\|M_p\|_{p,\sigma} \bigg|_{p=2}= \frac{1}{4}\|M\|_{2, \sigma}^{-1}\cdot \Big( \Ent_{2, \sigma}(M) + 4\tr\big[  \Gamma_\sigma^{\frac 12}(M'_2) \cdot \Gamma_\sigma^{\frac 12} (M) \big] \Big),$$
where
$$M'_2=\frac{\dd}{\dd p}M_p\bigg|_{p=2}=\frac{1}{4}\begin{pmatrix}
\|A\|_{2, \rho}^{-1}\cdot  \Ent_{2, \rho}(A) & w\\
w & \|B\|_{2, \rho}^{-1}\cdot  \Ent_{2, \rho}(B)
\end{pmatrix},$$
and $w =\|C\|_{2, \rho}^{-1}\cdot \left( \frac{1}{2}\Ent_{2, \rho}\big(I_{2, 2}(C)\big) +\frac{1}{2}\Ent_{2, \rho}\big(I_{2, 2}(C^{\dagger})\big)\right)$. We conclude that 
\begin{align*}
\frac{\dd}{\dd p}\|M_p\|_{p,\sigma} \bigg|_{p=2}= \frac{1}{4}\|M\|_{2, \sigma}^{-1}\cdot \Big(& \Ent_{2, \sigma}(M) + \theta \Ent_{2, \rho}(A) +(1-\theta) \Ent_{2, \rho}(B) \\ 
&\, +\sqrt{\theta(1-\theta)}\Ent_{2, \rho}\big(I_{2, 2}(C)\big)+\sqrt{\theta(1-\theta)}\Ent_{2, \rho}\big(I_{2, 2}(C^\dagger)\big) \Big).
\end{align*}
Comparing to~\eqref{eq:der-X} and using $\|M\|_{2, \sigma}=\|X\|_{2, \sigma\otimes \rho}$ the desired inequality follows.

\end{proof}

We need yet another lemma to prove Theorem~\ref{thm:tensorization-alpha-2}.

\begin{lemma}\label{lem:2-positive}
For any Lindblad generator $\cK$ that is $\rho$-reversible for some positive definite density matrix $\rho$ we have
$$\cE_{2, \cK}\big(I_{2, 2}(C)\big) +\cE_{2, \cK}\big(I_{2, 2}(C^\dagger)\big) \leq \langle C, \cK(C)\rangle_\rho+\langle C^\dagger, \cK(C^\dagger)\rangle_\rho.$$
for any $C$.
\end{lemma}

\begin{proof}
Define $D:=\Gamma_{\rho}^{\frac 12}(C)$. Then for $j\in \{0,1\}$
$$Y_{j}:= \begin{pmatrix}
|D| & (-1)^j D^\dagger\\
(-1)^j D & |D^\dagger|
\end{pmatrix}\geq 0,$$
is positive semidefinite~\cite{Bhatia15}. Since $\Gamma_{\rho}^{-1/2}$ is completely positive
we have
$$Z_{j}:=\id\otimes \Gamma_{\rho}^{-1/2}(Y_j) = \begin{pmatrix}
I_{2, 2}(C) & (-1)^j C^\dagger\\
(-1)^j C & I_{2, 2}(C^\dagger)
\end{pmatrix}\geq0.$$
On the other hand, $\Psi_t= \e^{-t\cK}$ is completely positive. Therefore,  
$$\id \otimes  \Psi_t (Z_0) = \begin{pmatrix}
 \Psi_t(I_{2, 2}(C)) &  \Psi_t(C^\dagger)\\
 \Psi_t(C) &  \Psi_t(I_{2, 2}(C^\dagger))
\end{pmatrix}\geq 0,$$
is positive semidefinite. Putting these together we find that 
$$g(t):=\langle Z_1, \id \otimes  \Psi_t (Z_0)\rangle_{\II\otimes \rho}\geq 0, \qquad \forall t\geq 0.$$
We note that 
\begin{align*}
g(t) & = \big\langle I_{2, 2}(C), \Psi_t(I_{2, 2}(C))\big\rangle_{\rho} + \big\langle I_{2, 2}(C^\dagger), \Psi_t(I_{2, 2}(C^\dagger))\big\rangle_{\rho} - \big\langle C, \Psi_t(C)\big\rangle_{\rho}- \big\langle C^\dagger, \Psi_t(C^\dagger)\big\rangle_{\rho}.
\end{align*} 
From this expression it is clear that 
$$g(0)= \|I_{2, 2}(C)\|_{2, \rho}^2 + \|I_{2, 2}(C^\dagger)\|_{2, \rho}^2 - \|C\|_{2, \rho}^2 -\|C^\dagger\|_{2, \rho}^2 =0.$$
Therefore, we must have $g'(0)\geq 0$ which is equivalent to the desired inequality.

\end{proof}

Now we have all the required tools for proving Theorem~\ref{thm:tensorization-alpha-2}. Indeed, we can prove a stronger statement out of which Theorem~\ref{thm:tensorization-alpha-2} is implied by a simple induction.

\begin{theorem}\label{thm:tensorization-alpha-2-stronger}
Let $\dim \cH=2$ and $\cL(X) = X-\tr(\sigma X) \II$ for some positive definite density matrix $\sigma\in \cD_+(\cH)$. Also let $\cK$ be a Lindblad generator associated to a primitive QMS that is reversible with respect to some positive definite state $\rho\in \cD_+(\cH')$. 
Then we have 
$$\alpha_2(\cL\otimes \cI' + \cI\otimes \cK)= \min\{\alpha_2(\cL), \, \alpha_2(\cK)\},$$
where $\cI$ and $\cI'$ denote the identity superoperators acting on $\cB(\cH)$ and $\cB(\cH')$ respectively.   

\end{theorem}

\begin{proof}
Let $\alpha=\min\{\alpha_2(\cL), \, \alpha_2(\cK)\}$. 
By restricting $X$ in the $2$-log-Sobolev inequality to be of the tensor product form and using
$$\Ent_{2, \sigma\otimes \rho}(Y\otimes Y') = \Ent_{2, \sigma}(Y)+\Ent_{2, \rho}(Y'),$$ 
we conclude that $\alpha_2(\cL\otimes \cI + \cI\otimes \cK)\leq \alpha$.  To prove the inequality in the other direction we need to show that  for any $X\in \cP(\cH\otimes \cH')$ we have
\begin{align}\label{eq:des-09a}
\alpha\, \Ent_{2, \sigma\otimes \rho} (X) \leq \cE_{2, \cL\otimes \cI' + \cI\otimes \cK} (X).
\end{align}
Assume, without loss of generality, that $\sigma$ is diagonal of the form~\eqref{eq:sigma-diag}, and that $X\in \cP(\cH\otimes \cH')$ has the block form~\eqref{eq:X-block}. Define $M$ by~\eqref{eq:M-block}. Then by Lemma~\ref{lem:entropy-inequality} we have
\begin{align*}
\Ent_{2, \sigma\otimes \rho}(X)\leq &~ \Ent_{2, \sigma}(M) + \theta \Ent_{2, \rho}(A) +(1-\theta)\Ent_{2, \rho}(B)\\
& ~+ \sqrt{\theta(1-\theta)}\,\Ent_{2, \rho}(I_{2, 2}(C)) + \sqrt{\theta(1-\theta)}\,\Ent_{2, \rho}(I_{2, 2}(C^\dagger)).
\end{align*}
On the other hand by the definition of $\alpha$ we have
$$\alpha\,\Ent_{2, \sigma}(M) \leq \cE_{2, \cL}(M),$$
and
$$\alpha\, \Ent_{2, \rho}(Y) \leq \cE_{2, \cK}(Y),$$
for all $Y\in \big\{ A, B, I_{2, 2}(C), I_{2, 2}(C^\dagger)  \big\}$.
Therefore, we have 
\begin{align}
\alpha\,\Ent_{2, \sigma\otimes \rho}(X)&\leq  \cE_{2, \cL}(M) + \theta \cE_{2, \cK}(A) +(1-\theta)\cE_{2, \cK}(B)\nonumber\\
& \quad+ \sqrt{\theta(1-\theta)}\,\cE_{2, \cK}(I_{2, 2}(C)) + \sqrt{\theta(1-\theta)}\,\cE_{2, \cK}(I_{2, 2}(C^\dagger))\nonumber\\
&\leq  \cE_{2, \cL}(M) + \theta \cE_{2, \cK}(A) +(1-\theta)\cE_{2, \cK}(B)\nonumber\\
& \quad+ \sqrt{\theta(1-\theta)}\,\langle C, \cK(C)\rangle + \sqrt{\theta(1-\theta)}\,\langle C^\dagger, \cK(C^\dagger)\rangle, \label{eq:7yua}
\end{align}
where in the second inequality we use Lemma~\ref{lem:2-positive}.
We now have
\begin{align*}
\cE_{2, \,\cL\otimes \cI' + \cI\otimes \cK} (X) & = \langle X, (\cL\otimes \cI' + \cI\otimes \cK)(X)\rangle_{\sigma\otimes \rho}\\
& = \langle X, \cL\otimes \cI'(X)\rangle_{\sigma\otimes \rho} + \Bigg\langle \begin{pmatrix}
A &  C\\
 C^\dagger & B
\end{pmatrix},  \begin{pmatrix}
\cK(A) &  \cK(C)\\
 \cK(C^\dagger) & \cK(B)
\end{pmatrix}\Bigg\rangle_{\sigma\otimes \rho}.
\end{align*}
We compute each term in the above sum separately. 
\begin{align*}
\big\langle X, \,\cL\,\otimes\, &\,\cI'(X)\big\rangle_{2, \sigma\otimes \rho}  \\
& = \Bigg\langle \begin{pmatrix}
A &  C\\
 C^\dagger & B
\end{pmatrix},  \begin{pmatrix}
(1-\theta)(A-B) &  C\\
 C^\dagger & \theta(B-A)
\end{pmatrix}\Bigg\rangle_{2, \sigma\otimes \rho} \\
& = \theta(1-\theta)\langle A, A-B\rangle_\rho +\theta(1-\theta)\langle B, B-A\rangle_\rho + 2\sqrt{\theta(1-\theta)}\langle C, C\rangle_\rho\\
& = \theta(1-\theta)\|A\|_{2,\rho}^2 +\theta(1-\theta)\|B\|_{2,\rho}^2 -2\theta(1-\theta)\langle A, B\rangle_\rho+ 2\sqrt{\theta(1-\theta)}\|C\|_{2, \rho}\\
& \geq \theta(1-\theta)\|A\|_{2,\rho}^2 +\theta(1-\theta)\|B\|_{2,\rho}^2 -2\theta(1-\theta)\|A\|_{2, \rho}\cdot \|B\|_{2, \rho}+ 2\sqrt{\theta(1-\theta)}\|C\|_{2, \rho}\\
& = \langle M, \cL (M)\rangle_\sigma\\
& = \cE_{2, \cL}(M).
\end{align*}
For the second term we compute
\begin{align*}
\Bigg\langle \begin{pmatrix}
A &  C\\
 C^\dagger & B
\end{pmatrix}, & \begin{pmatrix}
\cK(A) &  \cK(C)\\
 \cK(C^\dagger) & \cK(B)
\end{pmatrix}\Bigg\rangle_{\sigma\otimes \rho} \\
& = \theta \langle A, \,\cK (A)\rangle_\rho + (1-\theta) \langle B, \cK (B)\rangle_\rho  \\
&\quad + \sqrt{\theta(1-\theta)} \langle C, \cK (C)\rangle_\rho  + \sqrt{\theta(1-\theta)} \langle C^\dagger, \cK (C^\dagger)\rangle_\rho\\
& = \theta \cE_{2, \cK}(A) + (1-\theta) \cE_{2, \cK}(B)  \\
&\quad + \sqrt{\theta(1-\theta)} \langle C, \cK (C)\rangle + \sqrt{\theta(1-\theta)} \langle C^\dagger, \cK (C^\dagger)\rangle.
\end{align*}
Therefore, we have
\begin{align*}
\cE_{2, \cL\otimes \cI' + \cI\otimes \cK}(X)& \geq  \cE_{2, \cL}(M) +\theta \cE_{2, \cK}(A) + (1-\theta) \cE_{2, \cK}(B)\\
& \quad   + \sqrt{\theta(1-\theta)} \langle C, \cK (C)\rangle + \sqrt{\theta(1-\theta)} \langle C^\dagger, \cK (C^\dagger)\rangle.
\end{align*}
Comparing this to~\eqref{eq:7yua} we arrive at the desired inequality~\eqref{eq:des-09a}. 

\end{proof}

We now give the exact expression of the $2$-log-Sobolev constant of the simple Lindblad generator (in any dimension). {We recall that the case of the $1$-log-Sobolev constant was found in \cite{MFW16} (see also \cite{KT13} when $\sigma=\II/d$). } {The proof in our general setting is similar to the one of \cite{MFW16}. We however provide it in \Cref{proofofLSI2} for the sake of completeness.}

\begin{theorem}\label{thm:LSC-2-simple}
Let $\sigma\in \cD_+(\cH)$ be arbitrary and let $\cL(X) = X-\tr(\sigma X) \II$ be the simple Lindblad generator. Then we have 
\begin{align}\label{eq:exact-LSC-2-simple}
\alpha_2(\cL)= \frac{1-2s_{\min}(\sigma)}{\log\big(1/s_{\min}(\sigma)-1\big)},
\end{align}
where $s_{\min}(\sigma)$ is the minimum eigenvalue of $\sigma$.
\end{theorem}

We can now derive a tensorization-type result for a wide class of Lindblad generators. 
Let $\cL$ be a  $\sigma$-reversible and primitive Lindblad generator. 
Recall that the \emph{spectral gap} of $\cL$ is defined by
$$\lambda(\cL) = \inf_{X} \frac{\cE_{2, \cL}(X)}{\Var_\sigma(X)},$$
where $\Var_\sigma(X) =\langle X, X\rangle_\sigma - \langle X, \II\rangle_\sigma^2 =\|X\|_{2, \sigma}^2 -\langle X, \II\rangle_\sigma^2$, see e.g.~\cite{KT13}. Observe that $\Var_\sigma(X)$ is the squared length of the projection of $X$ onto the subspace orthogonal to $\II \in \cB(\cH)$ with respect to the inner product $\langle \cdot, \cdot\rangle_\sigma$. On the other hand, $\II$ is  the sole\footnote{This $0$-eigenvector is unique since $\cL$ is assumed to be primitive.} $0$-eigenvector of $\cL$ {up to a phase} which is self-adjoint with respect to this inner product. Therefore, $\lambda(\cL)$ is the minimum non-zero eigenvalue of~$\cL$. Note that, since $\cL$ {has positive spectral gap}, the Dirichlet form $\cE_{2, \cL}$ is non-negative, so $\lambda(\cL)>0$. Indeed, $\lambda(\cL)$ is really the spectral gap of $\cL$ above the zero eigenvalue.

The spectral gap satisfies the tensorization property, as shown below. Observe  that 
$$ \cK_n= \sum_{i=1}^n \widehat \cL_i,$$
is a sum of mutually commuting operators. Then the eigenvalues of $\cK_n$ are summations of eigenvalues of individual $\widehat \cL_i$'s. Since each $\widehat\cL_i$ is a tensor product of $\cL$ with some identity superoperator, the set of its eigenvalues is the same as that of $\cL$. Using these we conclude that 
\begin{align}\label{eq:lambda-tensorization}
\lambda(\cK_n) = \lambda(\cL), \qquad \forall n.
\end{align}

It is well-known that $\lambda(\cL)\geq \alpha_2(\cL)$ \cite{CM15,KT13}. The following corollary gives a lower bound on $\alpha_2(\cL)$ in terms of $\lambda(\cL)$.

\begin{corollary}\label{cor:2-LSC-general}
Let $\dim \cH=2$ and $\sigma\in \cD_+(\cH)$.
For any $\sigma$-reversible primitive Lindblad generator $\cL$ we have
$$\alpha_2(\cK_n) \geq \frac{1-2s_{\min}(\sigma)}{\log\big( 1/s_{\min}(\sigma)-1\big) } \lambda(\cL),$$
where $s_{\min}(\sigma)$ denotes the minimal eigenvalue of $\sigma$.
\end{corollary}
This corollary is a non-commutative version of Corollary A.4 of~\cite{DSC96} and gives a stronger bound compared to Corollary 6 of \cite{TPK14}. It would be interesting to compare this corollary with the result of King~\cite{King14} who generalized the hypercontractivity inequalities of~\cite{MO10} for the unital qubit depolarizing channel to all unital qubit quantum Markov semigroups. Here, having a bound on the 2-log-Sobolev constant of the $\sigma$-reversible generalized qubit depolarizing channel (and its tensorization property), we derive a bound on the 2-log-Sobolev constant of all qubit $\sigma$-reversible QMS.

\begin{proof}[Proof of \Cref{cor:2-LSC-general}]
Let $\cL'$ be the simple Lindblad generator that is $\sigma$-reversible,
and let $X\in \cP(\cH^{\otimes n})$ be arbitrary. Then by Theorem~\ref{thm:tensorization-alpha-2} and Theorem~\ref{thm:LSC-2-simple} we have
\begin{align}\label{eq:lambda-s-min-2}
\frac{1-2s_{\min}(\sigma)}{\log\big( 1/s_{\min}(\sigma)-1\big) } \, \Ent_{2, \sigma^{\otimes n}} \leq \sum_{i=1}^n \big\langle X, \widehat \cL'_i(X)\big\rangle_{\sigma^{\otimes n}}.
\end{align}
Let $\mathcal W_i\subset \cB(\cH^{\otimes n})$ be the subspace spanned by operators of the form $A_1\otimes \cdots \otimes A_n \in \cB(\cH^{\otimes n})$ with $A_i=\II\in \cB(\cH)$. In other words, $\mathcal W_i = \ker(\widehat \cL'_i)$. Then $\big\langle X, \widehat \cL'_i(X)\big\rangle_{\sigma^{\otimes n}}$ equals the squared length of the projection of $X$ onto $\mathcal W_i^{\perp}$. 
On the other hand, since $\cL$ is primitive and $\sigma$-reversible, we also have $\mathcal W_i=\ker \widehat \cL_i $ and $\mathcal W_i^{\perp}$ is invariant under $\widehat \cL_i$. Moreover, by definition $\lambda(\widehat \cL_i)$ is the minimum eigenvalue of $\widehat\cL_i$ restricted to $\mathcal W_i^{\perp}$ (i.e., the minimum non-zero eigenvalue). We conclude that
$$\lambda(\widehat\cL_i)\big\langle X, \widehat \cL'_i(X)\big\rangle_{\sigma^{\otimes n}}\leq \big\langle X, \widehat \cL_i(X)\big\rangle_{\sigma^{\otimes n}}.$$
On the other hand since $\widehat\cL_i$ equals the tensor product of $\cL$ with some identity superoperators, $\lambda(\widehat \cL_i) = \lambda(\cL)$. Therefore, 
$$\lambda(\cL)\big\langle X, \widehat \cL'_i(X)\big\rangle_{\sigma^{\otimes n}}\leq \big\langle X, \widehat \cL_i(X)\big\rangle_{\sigma^{\otimes n}}.$$
Using this in~\eqref{eq:lambda-s-min-2} we arrive at
$$\lambda(\cL)\frac{1-2s_{\min}(\sigma)}{\log\big( 1/s_{\min}(\sigma)-1\big) } \, \Ent_{2, \sigma^{\otimes n}} \leq \sum_{i=1}^n \big\langle X, \widehat\cL_i(X)\big\rangle_{\sigma^{\otimes n}} = \langle X, \cK_n(X)\rangle_{\sigma^{\otimes n}}.$$
This gives the desired bound on $\alpha_2(\cK_n)$.

\end{proof}

\begin{corollary}\label{cor:tensor}
Let $\dim \cH=2$ and $\sigma\in \cD_+(\cH)$.
Let $\cL$ be a $\sigma$-reversible primitive Lindblad generator. Then for any $1\leq q\leq p$ and $t\geq 0$ satisfying 
\begin{align*}
t\geq \frac{\log\big( 1/s_{\min}(\sigma)-1\big)}{4\lambda(\cL) \, \big(1-2s_{\min}(\sigma)\big)}\log\frac{p-1}{q-1},
\end{align*}
we have 
$\| \Phi_t^{\otimes n}(X)\|_{p, \sigma}\leq \|X\|_{q, \sigma}$ for all $X> 0$

\end{corollary}

%**************************************

\section{Application: second-order converses}\label{sec6}
One of the primary goals of information theory is to find optimal rates of
information-theoretic tasks. For instance, for the task of information transmission over a noisy channel, this optimal rate is the capacity. The latter is said to satisfy the \textit{strong converse property} if any attempt to transmit information at a rate higher than it fails with certainty in the limit of infinitely many uses of the channel. In this section, we show how reverse hypercontractivity inequalities can be used to derive finite sample size strong converse bounds in the tasks of asymmetric quantum hypothesis testing and classical communication through a classical-quantum channel.

\subsection{Quantum hypothesis testing}
Binary quantum hypothesis testing concerns the problem of discriminating between two different
quantum states, and is essential for various quantum information-processing protocols. 
Suppose that a party, Bob, receives a quantum system, with the knowledge that it is prepared either in the state $\rho$ (the null hypothesis) or in the state $\sigma$ (the alternative hypothesis) over a finite-dimensional Hilbert space ${\cal H}$.
His aim is to infer which hypothesis is true, i.e., which state the system is in. To do so he performs a measurement on the system that he receives. This is most generally described by a POVM $\{T,\II - T\}$ where $0 \le T \le \II$; When the measurement outcome is $T$ he infers that the state is $\rho$, and otherwise it is $\sigma$. Adopting the nomenclature from classical hypothesis testing, we refer to $T$ as a test. The probability that Bob correctly guesses the state to be $\rho$ is then equal to $\tr(T \rho)$, whereas his probability of correctly guessing the state to be $\sigma$ is
$\tr((\II -T)\sigma)$. Bob can erroneously infer the state to be $\sigma$ when it is actually $\rho$ or vice versa. The corresponding error probabilities are referred to as the Type~I error and Type~II error, respectively, and are given as follows:
$$\alpha(T) := \tr ((\II - T)\rho),~~~~~~~ \beta(T) := \tr (T \sigma),$$
Correspondingly, if multiple (say, $n$) identical copies of the system are available, and a test $T_n \in {\cal B}({\cal H}^{\otimes n})$ is performed on the $n$ copies, then the Type~I and Type~II errors are given by
$$\alpha_n(T_n) := \tr ((\II_n - T_n)\rho^{\otimes n}),~~~~~~~ \beta_n(T_n) := \tr (T_n \sigma^{\otimes n}),$$
where $\II_n$ denotes the identity operator in ${\cal B}({\cal H}^{\otimes n})$. There is a trade-off between the two error probabilities and there are various ways to optimize them. In the setting of \textit{asymmetric quantum hypothesis testing}, one minimizes the Type~II error under the constraint that the Type~I error stays below a threshold value $\eps \in (0,1)$. In this case one is interested in the following quantity
\begin{align}
\beta_{n, \eps} := \min \{\beta_n(T_n) \, : \, \alpha_n(T_n) \leq \eps, ~ 0\leq T_n\leq \II_n\},
\label{bne}
\end{align}
where the infimum is taken over all possible tests $T_n \in {\cal B}({\cal H}^{\otimes n})$. The quantum Stein lemma \cite{hiai1991,ON00}  states that
$$\lim_{n \to \infty} \left(-\frac{1}{n}\log \beta_{n, \eps}\right) = D(\rho||\sigma)~~~~~\forall \eps\in(0,1).$$

The \emph{asymptotic strong converse rate} $R_{sc}$ of the above quantum hypothesis testing problem is defined to be the smallest number $R$ such that if
$$\limsup_{n \to \infty} \frac{1}{n} \log \beta_n(T_n) \leq - R,$$
for some sequence of tests $\{T_n\}_{n \in {\mathbb{N}}}$, then 
$$\lim_{n \to \infty} \alpha_n(T_n) = 1.$$ 
This quantity has been shown to be equal to Stein's exponent $D(\rho||\sigma)$. 
In this section we are interested in obtaining a bound on the rate of convergence of $\alpha_n(T_n)$as a function of $n$, that is when Bob receives a finite number of identical copies 
of the quantum system. We use reverse hypercontractivity in order to obtain our bound. Before stating and proving the main theorem of this section, we recall the following important inequality that will be used in the proof.
\begin{lemma}[Araki-Lieb-Thirring inequality \cite{AL02,A90}]\label{ALT}
For any $A,B\in\cP(\cH)$, and $r\in[0,1]$,
\begin{align*}
	\tr(B^{r/2}A^r B^{r/2} )\le \tr(B^{1/2} A B^{1/2})^r.
\end{align*}	
\end{lemma}	

Our main result, from which a bound for the finite blocklength strong converse rate follows directly as a corollary, is given by \Cref{thm:QHT}.

\begin{theorem}\label{thm:QHT}
	Let $\rho,\sigma\in\cD_+(\cH)$ being faithful density matrices.\footnote{What we really need is that the supports of $\rho$ and $\sigma$ being the same (and not being the whole $\cH$) since in this case we may restrict everything to this support. 
} Then for any test $0\le T_n\le \II_n$, where $T_n \in {\cal B}({\cal H}^{\otimes n})$
\begin{align}\label{eq-state}
\log \tr(\sigma^{\otimes n} T_n ) &\geq -nD(\rho\|\sigma) - 2 
\sqrt{{n \|\sigma^{-1/2}\rho\sigma^{-1/2}\|_\infty \log\frac{1}{\tr(\rho^{\otimes n} T_n)}}}    +\log\tr(\rho^{\otimes n} T_n). 
\end{align}
\end{theorem}

\begin{proof}
{	The result follows by combining \Cref{thm:tensorization-alpha-1} and \Cref{RHolder}}. For simplicity of notation we will use $\sigma_n:=\sigma^{\otimes n}$ and $\rho_n:=\rho^{\otimes n}$.
Let $0\le p,q\le 1$ and let $t \geq 0$ be such that 
\begin{align}\label{cond1}
(1-p)(1-q) &=\e^{-t}.
\end{align}
Let ${\cal L}$ denote the generator of a generalized depolarizing semigroup $\{\Phi_t:\, t\geq 0\}$ with invariant state $\rho$, i.e., $\Phi_t(X)=\e^{-t} X + (1-\e^{-t}) \tr(\rho X) \II$. By Theorem~\ref{thm:tensorization-alpha-1} the $1$-log-Sobolev constants of this QMS and its tensor powers are lower bounded by $1/4$.	Then using \Cref{RHolderHC} for $Y=T_n$ and $X=\Gamma_{\rho_n}^{-1}(\sigma_n)$ we obtain
	\begin{align}\label{eq12}
	\tr\big(\sigma_n \Phi_t^{\otimes n}(T_n)\big)\ge \big\|\Gamma_{\rho_n}^{-1}(\sigma_n)\big\|_{p,\rho_n}\|T_n\|_{q,\rho_n}.
	\end{align}	
An application of the Araki-Lieb-Thirring inequality, \Cref{ALT}, with $A=\sigma_n$, $B=\rho_n^{(1-p)/p}$ and $r=p\in [0,1]$ leads to
\begin{align*}
	\big\|\Gamma_{\rho_n}^{-1}(\sigma_n)\big\|_{p,\rho_n}=\left[\tr\Big( \rho_n^{(1-p)/2p}\sigma_n\rho_n^{(1-p)/2p}\Big)^p\right]^{1/p}\ge \left[\tr\big( \rho_n^{1-p}\,\sigma_n^{p}\,\big)\right]^{1/p} =\exp\left( - D_{1-p}(\rho_n\|\sigma_n)    \right),
\end{align*}
where $$D_{1-p}(\rho\|\sigma):=\frac{-1}{p}\log\tr\left(\sigma^p\,\rho^{1-p}\right),$$
denotes the sandwiched $p$-R\'enyi divergence between $\rho$ and $\sigma$. A very similar application of \Cref{ALT} for $A=T_n$ and $B=\rho_n^{1/q}$ and $r=q\in [0,1]$ yields
\begin{align*}
	\|T_n\|_{q,\rho_n}=\left[\tr\big(\rho_n^{1/2q} T_n\rho_n^{1/2q} \big)^q\right]^{1/q}\ge \left[\tr\big(\rho_n T_n^q\big)\right]^{1/q}\ge \left[\tr\big(\rho_n T_n\big)\right]^{1/q},
\end{align*}	
where in the last inequality, we used that $0\le T_n\le \II$, so that $T_n^q\ge T_n$. Using the last two bounds in~\eqref{eq12}, we get
\begin{align*}
		\tr(\sigma_n \Phi_t^{\otimes n}(T_n))\ge \left[\tr(\rho_n T_n)\right]^{1/q}\exp\left(  -D_{1-p}(\rho_n\|\sigma_n)    \right).
\end{align*}
Taking the limit $p \to 0$ (and $q\to 1-\e^{-t}$) on both sides of the above inequality yields
\begin{align}\label{eq15}
\tr(\sigma_n \Phi_t^{\otimes n}(T_n)) &\geq 
 \left[\tr(\rho_n T_n)\right]^{1/(1-\e^{-t})}\exp\left(  -D(\rho_n\|\sigma_n)    \right).
\end{align}	

Let $\gamma:=\|\sigma^{-1/2}{\rho}\sigma^{-1/2}\|_{\infty}$ and define the superoperator $\Psi_t$ by
$$\Psi_t(X) = \e^{-t} X +\gamma(1-\e^{-t})\tr(\sigma X)\,\II.$$
Then by induction on $n$ it can be shown that $\Psi_t^{\otimes n} -\Phi_t^{\otimes n}$ is a completely positive superoperator. This is clear from definitions for $n=1$, and for every $Y\in \cP(\cH^{\otimes n}\otimes \cH')$, where $\cH'$ is an arbitrary Hilbert space, we have
\begin{align*}
\Psi_t^{\otimes n}\otimes \cI (Y) & = \big(\Psi^{\otimes (n-1)}\otimes \cI\otimes \cI \big) \big(\cI^{\otimes (n-1)}\otimes \Psi_t\otimes \cI (Y)\big)\\
& \geq \big(\Phi^{\otimes (n-1)}\otimes \cI\otimes \cI \big) \big(\cI^{\otimes (n-1)}\otimes \Psi_t\otimes \cI (Y)\big)\\
& =\big(\cI^{\otimes (n-1)}\otimes \Psi_t\otimes \cI \big) \big(\Phi^{\otimes (n-1)}\otimes \cI\otimes \cI (Y)\big)\\
& \geq\big(\cI^{\otimes (n-1)}\otimes \Phi_t\otimes \cI \big) \big(\Phi^{\otimes (n-1)}\otimes \cI\otimes \cI (Y)\big)\\
&= \Phi_t^{\otimes n}\otimes \cI (Y),
\end{align*}
where in the inequalities come from the induction hypothesis and the base of induction. Therefore, $\Psi_t^{\otimes n} -\Phi_t^{\otimes n}$ is a completely positive. On the other hand, for every $Y\in \cB(\cH^{\otimes n})$ we have 
$$\tr\big(\sigma_n\Psi_t^{\otimes n}(Y)\big) =  \big(\e^{-t}+  \gamma(1-\e^{-t})\big)^n\,\tr(\sigma_n Y).$$
This equation is immediate for $n=1$, and for arbitrary $n$ can be proven by first observing that it holds for $Y=Y_1\otimes\cdots \otimes Y_n$ being of a tensor product form, and then using linearity. Putting these together we arrive at
\begin{align*}
	\tr\big(\sigma_n  \Phi_t^{\otimes n}(T_n)\big)&\leq \tr\big(\sigma_n \Psi_t^{\otimes n}\,(T_n)\big )	\\
	&= \big(\e^{-t}+  \gamma(1-\e^{-t})\big)^n\,\tr(\sigma_n T_n).
\end{align*}
Next using the fact that $\gamma\geq 1$ {(which follows simply by taking the trace of the operator inequality $\rho\le\gamma\sigma$),}  the convexity of $h(x)=x^\gamma$ implies $(h(x)-h(1))/(x-1)\geq h'(1)$ for every $x\geq 1$. Therefore, $\e^{\gamma t}-1\ge \gamma(\e^t-1)$ for every $t\geq 0$, and $\e^{-t}+  \gamma(1-\e^{-t})\leq \e^{(\gamma-1)t}$. As a result
\begin{align}
	\tr\big(\sigma_n  \Phi_t^{\otimes n}(T_n)\big)&\leq
	\e^{(\gamma-1)nt}\,\tr(\sigma_n T_n).\label{eq14}
\end{align}
Then from~\reff{eq15} and~\reff{eq14} we get
\begin{align*}
	\left[\tr(\rho_n T_n)\right]^{1/(1-\e^{-t})}\exp\left(  -D(\rho_n\|\sigma_n)    \right)\le \e^{(\gamma-1)nt}\tr(\sigma_nT_n ).
\end{align*}	
Taking the logarithm of both sides yields
\begin{align}
\log \tr(\sigma_nT_n ) &\geq -D(\rho_n\|\sigma_n)  - 
(\gamma -1)nt + \frac{1}{1-\e^{-t}} \log \tr(\rho_n T_n)\nonumber\\
&\geq -D(\rho_n\|\sigma_n)  - \gamma nt + \left(1+ \frac{1}{t}\right) \log \tr(\rho_n T_n),
\label{eq:opt}
\end{align}
where the second inequality follows from $\e^t \geq 1+ t$ and
$$\frac{1}{1-\e^{-t}} = 1+ \frac{1}{\e^t-1}\leq 1+\frac 1 t.$$
Optimizing~\eqref{eq:opt} over the choice of $t$ yields  
$$t=\left(\frac{-\log \tr(\rho_nT_n)}{\gamma n } \right)^{1/2},$$
and we obtain the desired inequality
\begin{align*}
\log \tr(\sigma_nT_n ) &\geq -nD(\rho\|\sigma) - 2 \sqrt{-\gamma n\log \tr(\rho_nT_n)}+  \log \tr(\rho_n T_n).
\end{align*}
	
\end{proof}

\begin{remark}\label{remarkweak}
	{The bound found by the present reverse hypercontractivity technique is weaker than the one found in Equation (75) of \cite{Mosonyi2015}, which is in particular tight as $n\to\infty$. However, as opposed to \cite{Mosonyi2015}, the techniques developed in this paper have the particular advantage that they can be generalized to obtain strong converses in various problems of quantum network information theory (see \cite{cheng2019strong,cheng2019strongb}).
	}
	\end{remark}

\begin{corollary}[Finite-blocklength strong converse bound for quantum hypothesis testing]\label{corr-QHTstrong}
Let $\rho,\sigma\in\cD_+(\cH)$ and $\gamma=\|\rho\sigma^{-1}\|_\infty$. Then for any test $0\le T_n\le \II_n$, where $T_n \in {\cal B}({\cal H}^{\otimes n})$, if the Type~II error satisfies the inequality $\beta_n(T_n) \leq \e^{-nr}$ for $r > D(\rho||\sigma)$, then the Type~I error satisfies 
\begin{align}\label{eq-state}
\alpha_n(T_n) &\geq 1 - \e^{-nf},
\end{align}
where $$ f = \left( \sqrt{\gamma + (r-D(\rho||\sigma))} - \sqrt{\gamma}\right)^2,$$
and hence tends to zero in the limit of $r \to D(\rho||\sigma)$.
\end{corollary}
\begin{proof}
Fix $r> D(\rho\|\sigma)$ and consider a sequence of tests $T_n$ such that $\beta_n(T_n)\le \e^{-nr}$. Then, from \Cref{thm:QHT} we have 
\begin{align*}
-nr &\geq  -n D(\rho||\sigma) - 2\, \sqrt{n\gamma \log \frac{1}{1 - \alpha_n(T_n)}} - \log \frac{1}{1-\alpha_n(T_n)}. 
\end{align*}
Defining $x_n^2 := \log \frac{1}{1-\alpha_n(T_n)}$ this is equivalent to
$$x_n^2 + 2 \,\sqrt{n\gamma}\, x_n \,-\, n\, (r- D(\rho||\sigma)) \geq 0,$$
solving which directly leads to the statement of the corollary.
\end{proof}

\Cref{thm:QHT} also leads to the following finite blocklength second order lower bound on the Type~II error when the Type~I error is less than a threshold value.

\begin{corollary}\label{corr-soa}
Let $\rho,\sigma\in\cD_+(\cH)$ . Then for any $n \in {\mathbb{N}}$ 
and $\eps>0$ the minimal Type~II error satisfies 
\begin{align*}
\beta_{n, \eps} \geq (1-\eps) \exp\left( - nD(\rho||\sigma) - 2\, \sqrt{n\gamma \log\left(\frac{1}{1-\eps}\right)}\,\right),
\end{align*}	
where $\gamma= \|\rho\sigma^{-1}\|_\infty$.
\end{corollary}

\subsection{Classical-quantum channels}

The strong converse property of the capacity of a classical-quantum (c-q) channel was proved independently in \cite{796386,796385}. In this section, we use the quantum reverse hypercontractivity inequality to obtain a finite blocklength strong converse bound for transmission of information through classical-quantum (c-q) channels. Suppose Alice wants to send classical messages belonging to a finite set ${\cal M}$ to Bob, using a memoryless c-q channel:
$${\cal W}: {\cal X} \to {\cal D}({\cal H_B}),$$
where ${\cal X}$ denotes a finite alphabet, and ${\cal H_B}$ is a finite-dimensional Hilbert space with dimension $d$. Thus the output of the channel under input $x\in \mathcal X$ is some quantum state $\rho_x=\mathcal W(x)\in \cD(\cH_B)$.
To send a message $m \in {\cal M}$, Alice encodes it in a codeword
$${\cal E}^{(n)}(m) = x^n(m)\equiv x^n  := (x_1, x_2, \ldots x_n)  \in {\cal X}^n,$$ 
where ${\cal E}^{(n)}$ denotes the encoding map. She then sends it to Bob through $n$ successive uses of the channel ${\cal W}^{\otimes n}$, whose action on the codeword $x^n$ is given by
$${\cal W}^{\otimes n}(x^n) = \rho_{x_1} \otimes \cdots \otimes \rho_{x_n} \equiv \rho_{x^n} .$$
In order to infer Alice's message, Bob applies a measurement, described by a POVM $\Pi^n:= \{\Pi^n_{m'}\}_{m' \in {\cal M}}$ on the state ${\cal W}^{\otimes n}(x^n)=\rho_{x^n}$
that he receives. The outcome of the measurement would be Bob's guess of Alice's message. See Figure~\ref{fig:c-q-channel}. 

\FloatBarrier

\begin{figure}[t]
	\centering

\begin{tikzpicture}
\bXInput[$\mathcal{M}\ni m$]{source}
\bXBloc[3]{encoder}{$\cE^{(n)}$}{source}
\bXLink[]{source}{encoder}

\bXBloc[5]{decoder}{$\mathcal{W}^{\otimes n}$}{encoder}
\bXLink[$x^n\in\mathcal{X}^n$]{encoder}{decoder}

\bXBloc[11]{measure}{$\Pi^n:=\{\Pi^n_{m'}\}_{m'\in\mathcal{M}}$}{decoder}
\bXLink[$\rho_{x^n}=\rho_{x_1}\otimes\,...\,\otimes \rho_{x_n}$]{decoder}{measure}

\bXOutput[3]{output}{measure}
\bXLink[]{measure}{output}
\bXLinkName[0]{output}{$~~~~~~~~~~~~~\hat{m}\in\mathcal{M}$}

\end{tikzpicture}
\caption{Encoding and decoding of a classical message sent over a c-q channel. $\mathcal E^{(n)}$ is the encoding map, and $\Pi^n$ is the POVM constituting the decoding map.}
\label{fig:c-q-channel}
\end{figure}
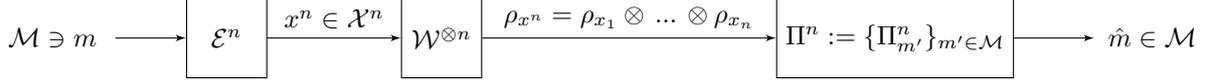

The triple $(|\mathcal{M}|,\mathcal{E}^{(n)},\Pi^n)$ defines a code which we denote as $\mathcal{C}_n$ (see \cite{watrous_2018}). The \emph{rate} of the code is given by $\log|\mathcal{M}|/n$, and its \emph{maximum probability of error} is given by
\begin{align*}
p_{\max}({\cal C}_n; {\cal W}):=\max_{m\in\mathcal{M}} \Big[1-\tr\big(\,\Pi^n_m\,\mathcal{W}^{\otimes n}\circ\cE^{(n)}(m)\big)\Big].
\end{align*}	
We let $C_{n, \eps}(\mathcal W)$ be the maximum rate $\log|\mathcal M|/n$ over all codes $\mathcal C_n=(|\mathcal M|, \mathcal E^{(n)}, \Pi^n)$ with $p_{\max}(\mathcal C_n; \mathcal W)\leq \eps$. Then
the (asymptotic) capacity of the channel is defined by 
 $$C(\mathcal{W}):=\lim_{\eps\to0}\liminf_{n\to\infty} C_{n, \eps}(\mathcal{W}).$$ 
For c-q channels, this is known to be given by \cite{schumacher1997sending,holevo1998capacity}
\begin{align*}
C(\mathcal W) = \max_{P_X} I(X; B)_\rho.
\end{align*}
Here the maximum is taken over all probability distributions $P_X$ on $\mathcal X$, the bipartite state $\rho_{XB}$ is given by 
$$\rho_{XB} = \sum_{x\in \mathcal X} P_X(x) \ket x\bra x\otimes \rho_x,$$
and $I(X; B)_\rho = D(\rho_{XB}\| \rho_X\otimes \rho_B)$ is the mutual information function. 
The fact that the capacity is given by maximum mutual information is indeed implied by its \emph{additivity}~\cite{Shor}. That is, the maximum mutual information associated to the channel $\mathcal W^{\otimes n}$ equals $n$ times the maximum mutual information of $\mathcal W$:
\begin{align}\label{eq:c-q-additivity}
\max_{P_{X^n}} I(X^n; B^n) = n \max_{P_X} I(X; B) = nC(\mathcal W).
\end{align}

\begin{theorem}\label{theo2}
 Let $\mathcal{W}:\mathcal{X}\to\cD(\cH_B)$ be a c-q channel with $\mathcal W(x)=\rho_x$ being faithful for all $x\in \mathcal X$. Then, for any code
${\cal C}_n:=(|\mathcal{M}|,\cE^{(n)},\Pi^n)$ with $p_{\max}(\mathcal C_n; \mathcal W)\leq \eps$ we have
\begin{align*}
I(X^n; B^n)\geq \log |\mathcal M| -2\sqrt{dn\log \frac{1}{1-\eps}} - \log \frac{1}{1-\eps}, 
\end{align*}	
where $d=\dim \cH_B$ and the mutual information is computed with respect to the state
$$\rho_{X^nB^n} = \frac{1}{|\mathcal M|}\sum_{m} \ket{x^n(m)}\bra{x^n(m)}\otimes \rho_{x^n(m)}.$$
\end{theorem}	

This theorem together with the additivity result~\eqref{eq:c-q-additivity}  directly imply that for any code of rate larger than $C(\mathcal{W})$, the maximum probability of error goes to one, as $n\to\infty$.	

\begin{proof}
For every $x^n=(x_1, \dots, x_n)\in \mathcal X^n$ let $\Phi_{t, x^n}=\Phi_{t, x_1}\otimes \cdots \otimes \Phi_{t, x_n}$ with
$$\Phi_{t, x}(X) = \e^{-t}X +(1-\e^{-t}) \tr(\rho_x X)\II.$$
Then following similar steps as in the proof of Theorem~\ref{thm:QHT}, using Theorem~\ref{thm:tensorization-alpha-1}, \Cref{RHolderHC} and the Araki-Lieb-Thirring inequality, for every $\Pi_m^n$ we have
$$\tr\big(\rho_{B^n} \Phi_{t, x^n}(\Pi^n_m)\big) \geq \big[\tr\big( \rho_{x^n} \Pi_m^n  \big)\big]^{1/(1-e^{-t})} \e^{-D(\rho_{x^n}\| \rho_{B^n})}.$$
Letting $x^n=x^n(m)$, using $\tr\big( \rho_{x^n(m)} \Pi_m^n  \big)\geq 1-\eps$, taking logarithm of both sides and averaging over the choice of $m\in \mathcal M$ we obtain
\begin{align*}
\frac{1}{|\mathcal M|} \sum_{m\in \mathcal M} \log \tr\big(\rho_{B^n} \Phi_{t, x^n(m)}(\Pi^n_m)\big) & \geq -\frac{1}{|\mathcal M|} \sum_{m\in \mathcal M} D(\rho_{x^n(m)} \| \rho_{B^n}) + \frac{1}{1-\e^{-t}} \log(1-\eps) \\
&= -I(X^n; B^n) + \frac{1}{1-\e^{-t}} \log(1-\eps)\\
&\geq -I(X^n; B^n) + \big(1+\frac{1}{t}\big) \log(1-\eps).
\end{align*}
Now define $\Psi_t(X) = \e^{-t}X + (1-\e^{-t})\tr(X)\II $. Following similar steps as in the proof of Theorem~\ref{thm:QHT}, using $\rho_{x} \leq \II$ it can be shown that $\Psi_t^{\otimes n} - \Phi_{t, x^n(m)}$ is completely positive. Therefore, $\Phi_{t, x^n(m)}(\Pi^n_m)\leq \Psi_t^{\otimes n}(\Pi_m^n)$ and we have  
\begin{align*}
-I(X^n; B^n) + \big(1+\frac{1}{t}\big) \log(1-\eps) & \leq \frac{1}{|\mathcal M|} \sum_m \log \tr\big(\rho_{B^n} \Psi_t^{\otimes n} (\Pi_m^n)\big)\\
&\leq \log \Big(\frac{1}{|\mathcal M|} \sum_m  \tr\big(\rho_{B^n} \Psi_t^{\otimes n} (\Pi_m^n)\Big)\\
& = \log   \Big( \frac{1}{|\mathcal M|}    \tr\big( \rho_{B^n} \Psi_t^{\otimes n}(\II^{\otimes n}_B)  \big)  \Big),
\end{align*}
where the second line follows from the concavity of the logarithm function and in the third line we use the fact that $\{\Pi^n_m:\, m\in\mathcal M\}$ is a POVM.
On the other hand, 
$$\Psi_t^{\otimes n}(\II^{\otimes n}_B) = \big(  \e^{-t} + (1-\e^{-t}) d   \big)^n \II^{\otimes n}_B\leq \e^{(d-1)nt}\II^{\otimes n}_B$$
Therefore, 
\begin{align*}
-I(X^n; B^n) + \big(1+\frac{1}{t}\big) \log(1-\eps) \leq -\log|\mathcal M|  + dnt.
\end{align*}
Optimizing over the choice of $t> 0$, the desired result follows.
\end{proof}
{The above theorem leads to the following finite blocklength second order strong converse bound for the classical capacity of a c-q channel.
\begin{corollary}	
	For any sequence of codes $\mathcal{C}_n:=(|\mathcal{M}|,\mathcal{E}^{(n)},\Pi^n)$ of rates $r:=\frac{|\mathcal{M}|}{n}>{C}(\mathcal{W})$, 
	\begin{align*}
	p_{\max}(\mathcal{C}_n;\mathcal{W})\ge 1-\e^{-nf}\,,
	\end{align*}
	where $f:=\big(\sqrt{d+(r-C(\mathcal{W}))}-\sqrt{d}  \big)^2$.
	\end{corollary}
\begin{proof}
We apply the bound found in \Cref{theo2}, so that
\begin{align*}
nC(\mathcal{W})\ge \log|\mathcal{M}|-2\sqrt{dn\log \frac{1}{1-\eps}} - \log \frac{1}{1-\eps}.
\end{align*}
The result follows by an analysis similar to the one of \Cref{corr-QHTstrong}.
	\end{proof}
\begin{remark}
	As pointed out in \Cref{remarkweak}, the strong converse bound that we find here is weaker than the one of \cite{Mosonyi2017}. However, and as opposed to \cite{Mosonyi2017}, our technique has recently been successfully applied to network information theoretical scenarios (see \cite{cheng2019strong,cheng2019strongb}).
	\end{remark}
}

%%%%%%%%%%%%%%%%%%%%%%%%%%%%%%%%%%%%
%*****************************APPENDIX**************************************

\appendix

\vspace{.5in}

\begin{center}
\textbf{\Large Appendix}
\end{center}

\section{Proof of Proposition~\ref{prop:contraction}} \label{app:contraction}
{\rm (i)} As mentioned in~\cite{DB14} (and explicitly worked out in~\cite{Beigi13})  for $p\geq 1$, contractivity  can be proven using the Riesz-Thorin interpolation theorem. So we focus on $p\in (-\infty, -1]\cup [1/2, 1)$. First let $p=-q\in (-\infty, -1]$, and $X> 0$. We note that 
$$\|\Phi_t(X)\|_{p, \sigma} = \|\Phi_t(X)^{-1}\|_{q, \sigma}^{-1}.$$
On the other hand, $\Phi_t$ is completely positive and unital, and $z\mapsto z^{-1}$ is operator convex. Therefore, by operator Jensen's inequality $\Phi_t(X^{-1})\geq \Phi_t(X)^{-1}$ and by the monotonicity of the norm we have $\|\Phi_t(X)^{-1}\|_{q,\sigma} \leq \|\Phi_t(X^{-1})\|_{q,\sigma}$. We conclude that 
 $$\|\Phi_t(X)\|_{p, \sigma} \geq \|\Phi_t(X^{-1})\|_{q,\sigma}^{-1} \geq \|X^{-1}\|_{q, \sigma}^{-1} = \|X\|_{p, \sigma},$$
where for the second inequality we use $q$-contractivity of $\Phi_t$ for $q\geq 1$. 

Now suppose that $p\in [1/2, 1)$. We note that {its H\"{o}lder conjugate} $\hat p\in (-\infty, -1]$, and that $\Phi_t$ is reverse $\hat p$-contractive. Then using H\"older's duality, for $X>0$ we have
\begin{align*}
\|\Phi_t(X)\|_p& = \inf_{Y>0: \|Y\|_{\hat p, \sigma}\geq 1} \langle Y, \Phi_t(X)\rangle_\sigma\\
& = \inf_{Y>0: \|Y\|_{\hat p, \sigma}\geq 1} \langle \widehat{\Phi}_t(Y), X\rangle_\sigma\\
& \geq \inf_{Z>0: \|Z\|_{\hat p, \sigma}\geq 1} \langle Z, X\rangle_\sigma\\
& = \|X\|_{p, \sigma}\,,
\end{align*}
{where $\widehat{\Phi}_t$ is the adjoint of $\Phi_t$ with respect to $\langle.,.\rangle_\sigma$, for each $t\ge 0$.}
Here the first equality follows from Lemma~\ref{lem:reversible}, and the inequality follows from the $\hat p$-contractivity of $\Phi_t$, i.e, $\|\Phi_t(Y)\|_{\hat p, \sigma} \geq \|Y\|_{\hat p, \sigma} \geq 1$.

\medskip
\noindent (ii) As worked out in~\cite{CMT15} this is an immediate consequence of the operator Jensen inequality.

%************************************************
\section{Second proof of Theorem~\ref{thm:QSVineq}}\label{app:qSV}

The proof is very similar to the one used in~\cite{BarEID17} to prove the strong $L_p$-regularity of the Dirichlet forms. Before stating the proof we need some definitions.

For a compact set $I$ we let $C(I)$ to be the Banach space of continuous, complex valued functions on $I$ (equipped with the supremum norm). Then the Banach space $C(I\times I)$ becomes a $*$-algebra when endowed with the natural involution $f\mapsto f^*$ with $f^*(x,y)=\overline{f(x,y)}$. Thus $C(I\times I)$ is a $C^*$-algebra.

We endow $\cB(\cH)$ with a Hilbert space structure by equipping it with the Hilbert-Schmidt inner product: 
$$\langle X, Y\rangle_{\HS}:= \tr(X^\dagger Y).$$ 
Fix $X,Y\in\cB_{sa}(\cH)$, 
and let $I$ be a compact interval containing the spectrum of both $X$ and $Y$. 
We define a $*$-representation $\pi_{X,Y}: C(I\times I)\rightarrow \cB\big(\cB(\cH)\big)$ that is  uniquely determined by its action on tensor products of functions as follows. For $f, g\in C(I)$ we define $\pi_{X, Y}(f\otimes g)\in \cB\big( \cB(\cH) \big)$ by
\begin{align*}
	\pi_{X,Y}(f\otimes g) (Z)=f(X) Zg(Y),\qquad Z\in\cB(\cH).
\end{align*}
	
The following lemma can be found in \cite{BarEID17} (see Lemma 4.2):

\begin{lemma}\label{lemma4.1}
$\pi_{XY}$ is a $*$-representation between $C^*$-algebras. That is,
\begin{itemize}
\item[{\rm(i)}] $\pi_{XY}(1)=\mathcal{I}$, where $1$ is the constant function on $I\times I$ equal to $1$.
\item[{\rm(ii)}] $\pi_{XY}(f^*g)=\pi_{XY}(f)^*\pi_{XY}(g)$ for all $f,g\in C(I\times I)$.
\item[{\rm(iii)}] If $f\in C(I\times I)$, is a non-negative function, then $\pi_{XY}(f)$ is a positive semi-definite operator on $\cB(\cH)$ for the Hilbert-Schmidt inner product, i.e., $\pi_{X, Y}(f)\in \cP\big( \cB(\cH) \big)$. 
\end{itemize}
\end{lemma}

Now, for any function $f\in C(I)$, define $\tilde{f}$ to be the function in $C(I\times I)$ defined by
\begin{align}\label{eqftilde}
	\tilde{f}(s,t)=\left\{\begin{aligned}
		&\frac{f(s)-f(t)}{s-t}\qquad s\ne t\\
		&f'(s)\qquad~~~~~~~~~ s=t.
\end{aligned}\right.
\end{align}
The following lemma, proved in \cite{BarEID17} (see Lemma 4.2), gives a generalization of the chain rule formula to a derivation. 

\begin{lemma}\label{lemma4.2}
Let $X, Y\in \cB_{sa}(\cH)$ and
let $I$ be a compact interval containing the spectrums of $X, Y$.
Let $f\in C(I)$ be a continuously differentiable function such that $f(0)=0$. Then for all $V\in \cB(\cH)$ we have
\begin{align*}
		Vf(Y)-f(X)V=\pi_{XY}(\tilde{f})(VY-XV),
\end{align*} 
where $\tilde{f}$ is defined by~\eqref{eqftilde}.
\end{lemma}

We can now prove the theorem. 	
By the result of~\cite{CM16} (an extension of Lemma~\ref{lem:choi}), there are superoperators $\partial_j:\cB(\cH)\rightarrow \cB(\cH)$ of the form
$$\partial_j(X)= [V_j, X] =  V_j X- XV_j,$$
where $V_j\in \cB(\cH)$, such that 
\begin{align}\label{eq3.3}
\langle X,\cL(Y)\rangle_\sigma = \sum_{j}\langle \partial_jX,\partial_j Y\rangle_{\sigma}.
\end{align}
Moreover, $V_j$'s are such that there are $\omega_j\geq 0$ with
$$\sigma V_j = \omega_j V_j \sigma.$$
Using the above equation one can show~\cite{BarEID17} that
\begin{align}\label{eq:220}
&\partial_j\big(I_{q,p}(X)\big)=\Gamma_\sigma^{-\frac{1}{q}}\bigg(V_j \Big(\Gamma_{\sigma}^{\frac{1}{p}}\big(\omega_j^{-\frac{1}{2p}}X\big)\Big)^{\frac{p}{q}}-\Big(\Gamma^{\frac{1}{p}}_\sigma\big(\omega_j^{\frac{1}{2p}}X\big)\Big)^{\frac{p}{q}}V_j\bigg).
\end{align}

For arbitrary $X> 0$ define  $Y_j:= \omega_j^{-1/4}\, \Gamma_{\sigma}^{\frac{1}{2}}(X)$ and
 ${Z}_j:= \omega_j^{1/4}\,\Gamma_{\sigma}^{\frac{1}{2}}(X)$. Using~\eqref{eq:220} we compute
\begin{align}
\mathcal{E}_{q,\mathcal{L}}\big(I_{q,2}(X)\big)&=\frac{q\hat q}{4}\big\langle I_{\hat{q},q}\big(I_{q,2}(X)\big),\mathcal{L}\big(I_{q,2}(X)\big)\big\rangle_\sigma\nonumber\\
			&=\frac{q\hat q}{4}\big\langle I_{\hat{q},2}(X),\mathcal{L}\big(I_{q,2}(X)\big)\big\rangle_\sigma\nonumber\\
			&=\frac{q\hat q}{4} \sum_{j} \langle \partial_j I_{\hat{q},2}(X),\partial_j I_{q,2}(X)\rangle_\sigma\label{eq:app-03}\\			
			&=\frac{q\hat q}{4} \sum_{j} \Big\langle \Gamma_\sigma^{-\frac{1}{\hat{q}}} \Big(V_j {Y_j}^{2/\hat{q}}-{Z}_j^{2/\hat{q}}V_j\Big), \Gamma^{-\frac{1}{q}} \Big(V_jY_j^{2/q}-{Z}_j^{2/q}V_j\Big)\Big\rangle_\sigma\label{eq:app-04}\\
			&=\frac{q\hat q}{4} \sum_{j} \Big\langle V_j {Y_j}^{2/\hat{q}}-{Z}_j^{2/\hat{q}}V_j, V_j{Y}_j^{2/q}-{Z}_j^{2/q}V_j\Big\rangle_{\HS}\nonumber\\
			&=\frac{q\hat q}{4}\sum_{j} \Big\langle \pi_{{Z}_j,{Y}_j}\big(\tilde{f}_{2/\hat{q}}\big)(V_j{Y}_j-{Z}_jV_j),\pi_{{Z}_j,{Y}_j}\big(\tilde{f}_{2/q}\big)(V_j{Y}_j-{Z}_jV_j)\Big\rangle_{\HS}\label{eq:app-06}\\
			&=\frac{q\hat q}{4}\sum_{j} \Big\langle V_j{Y}_j-{Z}_jV_j, \pi_{{Z}_j,{Y}_j}\big(\tilde{f}_{2/\hat{q}}\big)^*\pi_{{Z}_j,{Y}_j}\big(\tilde{f}_{2/q}\big)(V_j{Y}_j-{Z}_jV_j)\Big\rangle_{\HS}\nonumber\\
			&= \frac{q\hat q}{4} \sum_{j}  \Big\langle V_j {Y}_j-Z_jV_j,\pi_{Z_j,Y_j}\big(\tilde{f}_{2/\hat{q}}^*\tilde{f}_{2/q}\big)(V_jY_j-Z_jV_j)\Big\rangle_{\HS}, \label{eq:app-08}
\end{align}
where in~\eqref{eq:app-03} we used~\eqref{eq3.3}, in~\eqref{eq:app-04} we used~\eqref{eq:220}, and in~\eqref{eq:app-06} we used the chain rule formula of Lemma \ref{lemma4.2} for the functions $f_\alpha$ with $f_\alpha(x)=x^{\alpha}$. Finally, in~\eqref{eq:app-08} we used part (ii) of Lemma~\ref{lemma4.1}. 

Now, using the proofs of Theorem~2.1 and Lemma~2.4 of~\cite{MOS12}, for any $x,y\geq 0$ and $0\leq p\leq q\leq 2$ we have 
\begin{align}\label{classineq}
q\hat{q}(x^{1/\hat{q}}-y^{1/\hat{q}})(x^{1/q}-y^{1/q})\le p\hat{p}(x^{1/\hat{p}}-y^{1/\hat{p}})(x^{1/p}-y^{1/p}). 
\end{align}
This means that for all $x, y$ we have 
$$q\hat q \big(\tilde{f}_{2/\hat{q}}^*\tilde{f}_{2/q}\big)(x, y) \leq p\hat p \big(\tilde{f}_{2/\hat{p}}^*\tilde{f}_{2/p}\big)(x, y).$$
Hence, by part (iii) of Lemma~\ref{lemma4.1} we have 
\begin{align*}
\mathcal{E}_{q,\mathcal{L}}(I_{q,2}(X))&\le \frac{p \hat{p}}{4} \sum_{j}  \Big\langle V_j Y_j-Z_jV_j,\pi_{Z_j,Y_j}(\tilde{f}_{2/\hat{p}}^*\tilde{f}_{2/p})(V_jY_j-Z_jV_j)\Big\rangle_{\HS}\\
			&=\mathcal{E}_{p,\mathcal{L}}(I_{p,2}(X)).
\end{align*}

{
\begin{remark}
	The difference with the proof of $L_p$-regularity of \cite{BarEID17} lies in the choice of the inequality \reff{classineq} used at the end of the proof.
	\end{remark}}

\section{Proof of  \Cref{thm:LSC-2-simple}}\label{proofofLSI2}
	Since both $\Ent_{2, \sigma}(X)$ and $\cE_{2, \cL}(X)$ are homogenous of degree two {in $X$}, to  prove a log-Sobolev inequality, without loss of generality we can assume that $X$ is of the form $X=\Gamma_\sigma^{-1/2}(\sqrt \rho)$ where $\rho$ is a density matrix. In this case 
	$$\Ent_{2, \sigma}(X) = D(\rho\| \sigma), \qquad \langle X, \cL X\rangle_\sigma = 1- \big[\tr\big(\sqrt \sigma \sqrt \rho\big)\big]^2.$$
	Let   $\sigma= \sum_{i=1}^d s_i\ket i\bra i$ and $\rho=\sum_{k=1}^d r_k \ket{\tilde k}\bra{\tilde k}$ be the eigen-decompositions of $\sigma$ and $\rho$. Then
	$$\Ent_{2, \sigma}(X) = \sum_{k=1}^d r_k \log r_k - \sum_{i, k=1}^d  |\langle i| \tilde k\rangle |^2r_k \log s_i,$$
	and 
	$$\langle X, \cL X\rangle_\sigma = 1- \Big( \sum_{i, k=1}^d |\langle i| \tilde k\rangle|^2 \sqrt{s_i r_k}     \Big)^2.$$
	Let $A=(a_{ik})_{d\times d}$ be a $d\times d$ matrix whose entries are given by
	$$a_{ik} = |\langle i| \tilde k\rangle |^2.$$ 
	Observe that, fixing the eigenvalues $s_i$'s and $r_k$'s, the entropy $\Ent_{2, \sigma}(X)$ is a linear function of $A$ and $\cE_{2, \cL}(X)$ is concave function of $A$. 
	On the other hand, 
	since both $\{\ket 1, \dots, \ket{d}\}$ and $\{\ket {\tilde 1}, \dots, \ket{\tilde d}\}$ form orthonormal bases, $A$ is a doubly stochastic matrix.  Then by Birkhoff's
	theorem, $A$ can be written as a convex combination of permutations matrices. We conclude that if an inequality of the form 
	\begin{align*}
	\beta  \Big(\sum_{k=1}^d r_k \log r_k - \sum_{i, k=1}^d  a_{ik}r_k \log s_i\Big)\leq 1- \Big( \sum_{i, k=1}^d a_{ik} \sqrt{s_i r_k}     \Big)^2,
	\end{align*}
	holds for all permutation matrices $A$, then it holds for all doubly stochastic $A$, and then for all $\sigma, \rho$ with the given eigenvalues.  We note that $A$ is a permutation matrix when $\{\ket 1, \dots, \ket{d}\}$ and $\{\ket {\tilde 1}, \dots, \ket{\tilde d}\}$ are the same bases ({up to} some permutation) which means that $\sigma$ and $\rho$ commute. Therefore, a log-Sobolev inequality of the form 
	$$\beta \Ent_{2, \sigma} \big(  \Gamma_{\sigma}^{-1/2}(\rho) \big) \leq \cE_{2, \cL}\big(\Gamma_\sigma^{-1/2}(\rho)\big),$$
	holds for all $\rho$ if and only if it holds for all $\rho$ that commute with $\sigma$. That is, to find the log-Sobolev constant 
	$$\alpha_2(\cL) = \inf_{\rho} \frac{\cE_{2, \cL}\big(\Gamma_\sigma^{-1/2}(\rho)\big)}{\Ent_{2, \sigma} \big(  \Gamma_{\sigma}^{-1/2}(\rho) \big) },$$
	we may restrict to those $\rho$ that commute with $\sigma$. This optimization problem over such $\rho$ is equivalent to computing the $2$-log-Sobolev constant of the \emph{classical} simple Lindblad generator, and has been solved in Theorem~A.1 of~\cite{DSC96}.

\qed

%*******************************************************************************************
\bibliographystyle{abbrv}
\bibliography{biblio}

\begin{thebibliography}{10}

\bibitem{AG76}
R.~{Ahlswede} and P.~{Gacs}.
\newblock Spreading of sets in product spaces and hypercontraction of the
  {M}arkov operator.
\newblock {\em Ann. Probab.}, 4(6):925--939, 1976.

\bibitem{A90}
H.~Araki.
\newblock On an inequality of {L}ieb and {T}hirring.
\newblock {\em Letters in Mathematical Physics}, 19(2):167--170, Feb 1990.

\bibitem{BarEID17}
I.~Bardet.
\newblock {Estimating the decoherence time using non-commutative Functional
  Inequalities}.
\newblock {\em arXiv preprint arXiv:1710.01039}, 2017.

\bibitem{bardet2018hypercontractivity}
I.~Bardet and C.~Rouz{\'e}.
\newblock Hypercontractivity and logarithmic {S}obolev inequality for
  non-primitive quantum markov semigroups and estimation of decoherence rates.
\newblock {\em arXiv preprint:1803.05379}, 2018.

\bibitem{Beigi13}
S.~Beigi.
\newblock Sandwiched {R}\'enyi divergence satisfies data processing inequality.
\newblock {\em Journal of Mathematical Physics}, 54:122202, 2013.

\bibitem{BK16}
S.~Beigi and C.~King.
\newblock {Hypercontractivity and the logarithmic Sobolev inequality for the
  completely bounded norm}.
\newblock {\em Journal of Mathematical Physics}, 57(1):015206, 2016.

\bibitem{Bhatia15}
R.~Bhatia.
\newblock {\em Positive Definite Matrices}.
\newblock Princeton Series in Applied Mathematics. Princeton University Press,
  2015.

\bibitem{BLM13}
S.~{Boucheron}, G.~{Lugosi}, and P.~{Massart}.
\newblock {\em Concentration Inequalities: A Nonasymptotic Theory of
  Independence}.
\newblock Oxford University Press, 2013.

\bibitem{capel2018quantum}
A.~Capel, A.~Lucia, and D.~P{\'e}rez-Garc{\'\i}a.
\newblock Quantum conditional relative entropy and quasi-factorization of the
  relative entropy.
\newblock {\em Journal of Physics A: Mathematical and Theoretical},
  51(48):484001, 2018.

\bibitem{CM15}
R.~Carbone and A.~Martinelli.
\newblock Logarithmic {S}obolev inequalities in non-commutative algebras.
\newblock {\em Infinite Dimensional Analysis, Quantum Probability and Related
  Topics}, 18(02):1550011, 2015.

\bibitem{CM16}
E.~A. Carlen and J.~Maas.
\newblock Gradient flow and entropy inequalities for quantum markov semigroups
  with detailed balance.
\newblock {\em Journal of Functional Analysis}, 273(5):1810 -- 1869, 2017.

\bibitem{cheng2019strong}
H.-C. Cheng, N.~Datta, and C.~Rouz{\'e}.
\newblock Strong converse bounds in quantum network information theory:
  distributed hypothesis testing and source coding.
\newblock {\em arXiv preprint arXiv:1905.00873}, 2019.

\bibitem{cheng2019strongb}
H.-C. Cheng, N.~Datta, and C.~Rouz{\'e}.
\newblock Strong converse for classical-quantum degraded broadcast channels.
\newblock {\em arXiv preprint arXiv:1905.00874}, 2019.

\bibitem{CMT15}
T.~Cubitt, M.~Kastoryano, A.~Montanaro, and K.~Temme.
\newblock {Quantum reverse hypercontractivity}.
\newblock {\em Journal of Mathematical Physics}, 56(10), 2015.

\bibitem{deWolf08}
R.~de~Wolf.
\newblock A brief introduction to {F}ourier analysis on the {B}oolean cube.
\newblock {\em Theory of Computing}, 1:1--20, 2008.

\bibitem{DB14}
P.~Delgosha and S.~Beigi.
\newblock Impossibility of local state transformation via hypercontractivity.
\newblock {\em Communications in Mathematical Physics}, 332(1):449--476, 2014.

\bibitem{DJKR16}
I.~Devetak, M.~Junge, C.~King, and M.~B. Ruskai.
\newblock {Multiplicativity of completely bounded $p$-norms implies a new
  additivity result}.
\newblock {\em Communications in Mathematical Physics}, 266(1):37--63, 2006.

\bibitem{DSC96}
P.~Diaconis and L.~Saloff-Coste.
\newblock {Logarithmic Sobolev inequalities for finite Markov chains}.
\newblock {\em The Annals of Applied Probability}, 6(3):695--750, 1996.

\bibitem{FL13}
R.~L. Frank and E.~H. Lieb.
\newblock Monotonicity of a relative r\'enyi entropy.
\newblock {\em Journal of Mathematical Physics}, 54:122201, 2013.

\bibitem{GKS76}
V.~Gorini, A.~Kossakowski, and E.~C.~G. Sudarshan.
\newblock {Complete positive dynamical semigroups of N-level systems}.
\newblock {\em Journal of Mathematical Physics}, 17(1976):821, 1976.

\bibitem{GL10}
N.~{Gozlan} and C.~{Leonard}.
\newblock Transport inequalities. a survey.
\newblock {\em Markov Processes and Related Fields}, 16:635--736, 2010.

\bibitem{G75b}
L.~Gross.
\newblock Logarithmic sobolev inequalities.
\newblock {\em American Journal of Mathematics}, 97(4):1061--1083, 1975.

\bibitem{hiai1991}
F.~Hiai and D.~Petz.
\newblock The proper formula for relative entropy and its asymptotics in
  quantum probability.
\newblock {\em Communications in Mathematical Physics}, 143(1):99--114, 1991.

\bibitem{holevo1998capacity}
A.~S. Holevo.
\newblock The capacity of the quantum channel with general signal states.
\newblock {\em IEEE Transactions on Information Theory}, 44(1):269--273, 1998.

\bibitem{KA12}
S.~{Kamath} and V.~{Anantharam}.
\newblock Non-interactive simulation of joint distributions: The
  {H}irschfeld--{G}ebelein--{R}\'enyi maximal correlation and the
  hypercontractivity ribbon.
\newblock In {\em Proc. 50th Ann. Allerton Conf. Commun., Control Comput.},
  pages 1057--1064, 2012.

\bibitem{KT13}
M.~J. Kastoryano and K.~Temme.
\newblock {Quantum logarithmic Sobolev inequalities and rapid mixing}.
\newblock {\em Journal of Mathematical Physics}, 54(5), 2013.

\bibitem{King03}
C.~King.
\newblock Inequalities for trace norms of $2\times 2$ block matrices.
\newblock {\em Communications in Mathematical Physics}, 242(3):531--545,
  November 2003.

\bibitem{King14}
C.~King.
\newblock Hypercontractivity for semigroups of unital qubit channels.
\newblock {\em Communications in Mathematical Physics}, 328(1):285--301, May
  2014.

\bibitem{AL02}
E.~H. Lieb and W.~E. Thirring.
\newblock Inequalities for the moments of the eigenvalues of the {S}chrodinger
  {H}amiltonian and their relation to {S}obolev inequalities.
\newblock In {\em The Stability of Matter: From Atoms to Stars}, pages
  135--169. Springer, 1991.

\bibitem{Lind}
G.~Lindblad.
\newblock {On the generators of quantum dynamical semigroups}.
\newblock {\em Comm. Math. Phys.}, 48(2):119--130, 1976.

\bibitem{LHV17}
J.~Liu, R.~van Handel, and S.~Verd{\'u}.
\newblock Beyond the blowing-up lemma: Sharp converses via reverse
  hypercontractivity.
\newblock In {\em 2017 IEEE International Symposium on Information Theory
  (ISIT)}, pages 943--947, June 2017.

\bibitem{M12}
A.~Montanaro.
\newblock {Some applications of hypercontractive inequalities in quantum
  information theory}.
\newblock {\em Journal of Mathematical Physics}, 53(12):1--18, 2012.

\bibitem{MO10}
A.~Montanaro and T.~J. Osborne.
\newblock Quantum boolean functions.
\newblock {\em Chicago Journal of Theoretical Computer Science}, 2010(1), 2010.

\bibitem{Mosonyi2015}
M.~Mosonyi and T.~Ogawa.
\newblock Quantum hypothesis testing and the operational interpretation of the
  quantum {R}{\'e}nyi relative entropies.
\newblock {\em Communications in Mathematical Physics}, 334(3):1617--1648, Mar
  2015.

\bibitem{Mosonyi2017}
M.~Mosonyi and T.~Ogawa.
\newblock Strong converse exponent for classical-quantum channel coding.
\newblock {\em Communications in Mathematical Physics}, 355(1):373--426, Oct
  2017.

\bibitem{MOS12}
E.~Mossel, K.~Oleszkiewicz, and A.~Sen.
\newblock On reverse hypercontractivity.
\newblock {\em Geometric and Functional Analysis}, 23(3):1062--1097, Jun 2013.

\bibitem{MDSFT13}
M.~M\"uler-Lennert, F.~Dupuis, O.~Szehr, S.~Fehr, and M.~Tomamichel.
\newblock On quantum {R}\'{e}nyi entropies: a new generalization and some
  properties.
\newblock {\em Journal of Mathematical Physics}, 54(12):122203, 2013.

\bibitem{MFW16}
A.~M{\"u}ller-Hermes, D.~S. Fran{\c c}a, and M.~M. Wolf.
\newblock Relative entropy convergence for depolarizing channels.
\newblock {\em Journal of Mathematical Physics}, 57(2):022202, 2016.

\bibitem{MSFW}
A.~{M{\"u}ller-Hermes}, D.~{Stilck Franca}, and M.~M. {Wolf}.
\newblock {Entropy Production of Doubly Stochastic Quantum Channels}.
\newblock {\em J. Math. Phys.}, 57(2):022203, 2016.

\bibitem{N66}
E.~Nelson.
\newblock A quartic interaction in two dimensions.
\newblock {\em Mathematical theory of elementary particles}, pages 69--73,
  1966.

\bibitem{796386}
T.~Ogawa and H.~Nagaoka.
\newblock Strong converse to the quantum channel coding theorem.
\newblock {\em IEEE Transactions on Information Theory}, 45(7):2486--2489, Nov
  1999.

\bibitem{ON00}
T.~Ogawa and H.~Nagaoka.
\newblock Strong converse and {S}tein's lemma in quantum hypothesis testing.
\newblock {\em IEEE Transactions on Information Theory}, 46(7):2428--2433, Nov
  2000.

\bibitem{OZ99}
R.~Olkiewicz and B.~Zegarlinski.
\newblock {Hypercontractivity in noncommutative {$L_p$} spaces}.
\newblock {\em J. Funct. Anal.}, 161(1):246--285, 1999.

\bibitem{Petz88}
D.~{Petz}.
\newblock A variational expression for the relative entropy.
\newblock {\em Communications in Mathematical Physics}, 114(2):345--349, 1988.

\bibitem{RS13}
M.~Raginsky and I.~Sason.
\newblock Concentration of measure inequalities in information theory,
  communications, and coding.
\newblock {\em Foundations and Trends{\textregistered} in Communications and
  Information Theory}, 10(1-2):1--246, 2013.

\bibitem{schumacher1997sending}
B.~Schumacher and M.~D. Westmoreland.
\newblock Sending classical information via noisy quantum channels.
\newblock {\em Physical Review A}, 56(1):131, 1997.

\bibitem{Shor}
P.~W. Shor.
\newblock Additivity of the classical capacity of entanglement-breaking quantum
  channels.
\newblock {\em Journal of Mathematical Physics}, 43(9):4334--4340, 2002.

\bibitem{SH-K72}
B.~{Simon} and R.~{Hoegh-Krohn}.
\newblock Hypercontractive semigroups and two dimensional self-coupled bose
  fields.
\newblock {\em J. Funct. Anal.}, 9(2):121--180, 1972.

\bibitem{TPK14}
K.~{Temme}, F.~{Pastawski}, and M.~J. {Kastoryano}.
\newblock {Hypercontractivity of quasi-free quantum semigroups}.
\newblock {\em Journal of Physics A Mathematical General}, 47:5303, Oct. 2014.

\bibitem{TBH14}
M.~Tomamichel, M.~Berta, and M.~Hayashi.
\newblock Relating different quantum generalizations of the conditional
  r{\'e}nyi entropy.
\newblock {\em Journal of Mathematical Physics}, 55(8):082206, 2014.

\bibitem{watrous_2018}
J.~Watrous.
\newblock {\em The Theory of Quantum Information}.
\newblock Cambridge University Press, 2018.

\bibitem{WWY14}
M.~M. Wilde, A.~Winter, and D.~Yang.
\newblock Strong converse for the classical capacity of entanglement-breaking
  channels.
\newblock {\em Communications in Mathematical Physics}, 331(2):593--622,
  October 2014.

\bibitem{796385}
A.~Winter.
\newblock Coding theorem and strong converse for quantum channels.
\newblock {\em IEEE Transactions on Information Theory}, 45(7):2481--2485, Nov
  1999.

\bibitem{wolftour}
M.~M. Wolf.
\newblock {Quantum channels \& operations: Guided tour}.
\newblock
  \url{http://www-m5.ma.tum.de/foswiki/pub/M5/Allgemeines/MichaelWolf/QChannelLecture.pdf},
  2012.
\newblock Lecture notes based on a course given at the Niels-Bohr Institute.

\end{thebibliography}

\end{document}